\documentclass[11pt]{article}
\pdfoutput=1
\usepackage[protrusion=true,expansion=true]{microtype}
\usepackage{fullpage}
\usepackage{amsmath,amssymb,amsfonts,amsthm}
\usepackage{subcaption}
\usepackage{graphicx}
\usepackage[backref=page]{hyperref}
\usepackage{color}
\usepackage{wrapfig}
\usepackage{tikz}
\usetikzlibrary{decorations.pathreplacing}
\usepackage{setspace}
\usepackage{algorithm}
\usepackage[noend]{algpseudocode}
\usepackage[framemethod=tikz]{mdframed}
\usepackage{xspace}
\usepackage{pgfplots}
\usepackage{framed}
\usepackage{thmtools}
\usepackage{thm-restate}
\pgfplotsset{compat=1.5}

\newtheorem{theorem}{Theorem}[section]

\newtheorem{lemma}[theorem]{Lemma}

\newtheorem{definition}[theorem]{Definition}

\newtheorem{observation}[theorem]{Observation}

\newenvironment{proofof}[1]{\begin{trivlist} \item {\bf Proof
#1:~~}}
  {\qed\end{trivlist}}

\newcommand{\namedref}[2]{\hyperref[#2]{#1~\ref*{#2}}}
\newcommand{\thmlab}[1]{\label{thm:#1}}
\newcommand{\thmref}[1]{\namedref{Theorem}{thm:#1}}
\newcommand{\lemlab}[1]{\label{lem:#1}}
\newcommand{\lemref}[1]{\namedref{Lemma}{lem:#1}}

\newcommand{\seclab}[1]{\label{sec:#1}}
\newcommand{\secref}[1]{\namedref{Section}{sec:#1}}
\newcommand{\applab}[1]{\label{app:#1}}
\newcommand{\appref}[1]{\namedref{Appendix}{app:#1}}

\newcommand{\figlab}[1]{\label{fig:#1}}
\newcommand{\figref}[1]{\namedref{Figure}{fig:#1}}
\newcommand{\alglab}[1]{\label{alg:#1}}
\renewcommand{\algref}[1]{\namedref{Algorithm}{alg:#1}}

\newcommand{\deflab}[1]{\label{def:#1}}
\newcommand{\defref}[1]{\namedref{Definition}{def:#1}}

\newcommand{\obslab}[1]{\label{obs:#1}}

\def \ams    {\mdef{\mathsf{AMS}}}
\def \Lap    {\mdef{\mathsf{Lap}}}
\def \countsketch    {\mdef{\textsc{CountSketch}}}
\def \privcountsketch    {\mdef{\textsc{PrivCountSketch}}}

\newcommand{\PPr}[1]{\ensuremath{\mathbf{Pr}\left[#1\right]}}

\newcommand{\Ex}[1]{\ensuremath{\mathbb{E}\left[#1\right]}}

\renewcommand{\O}[1]{\ensuremath{\mathcal{O}\left(#1\right)}}
\newcommand{\tO}[1]{\ensuremath{\tilde{\mathcal{O}}\left(#1\right)}}
\newcommand{\eps}{\varepsilon}

\def \frakU    {\mdef{\mathfrak{U}}}
\def \frakS    {\mdef{\mathfrak{S}}}
\def \calA    {\mdef{\mathcal{A}}}

\def \calM    {\mdef{\mathcal{M}}}

\def \calY    {\mdef{\mathcal{Y}}}

\newcommand{\mdef}[1]{{\ensuremath{#1}}\xspace}  

\DeclareMathOperator*{\polylog}{polylog}
\DeclareMathOperator*{\poly}{poly}
\DeclareMathOperator*{\mc}{mc}
\DeclareMathOperator*{\mmc}{mmc}
\DeclareMathOperator*{\tail}{tail}
\DeclareMathOperator*{\median}{median}

\newcommand{\ceil}[1]{\mdef{\left\lceil#1\right\rceil}}               

\newcommand{\ignore}[1]{}

\newif\ifnotes\notestrue 
\ifnotes
\newcommand{\samson}[1]{\textcolor{purple}{{\bf (Samson:} {#1}{\bf ) }} \marginpar{\tiny\bf
             \begin{minipage}[t]{0.5in}
               \raggedright S:
            \end{minipage}}}            							
\else
\newcommand{\samson}[1]{}
\fi

\makeatletter
\renewcommand*{\@fnsymbol}[1]{\textcolor{darkpastelgreen}{\ensuremath{\ifcase#1\or *\or \dagger\or \ddagger\or
 \mathsection\or \triangledown\or \mathparagraph\or \|\or **\or \dagger\dagger
   \or \ddagger\ddagger \else\@ctrerr\fi}}}
\makeatother

\providecommand{\email}[1]{\href{mailto:#1}{\nolinkurl{#1}\xspace}}

\definecolor{ceruleanblue}{rgb}{0.16, 0.32, 0.75}
\definecolor{darkmidnightblue}{rgb}{0.0, 0.2, 0.4}
\definecolor{darkpastelgreen}{rgb}{0.01, 0.75, 0.24}
\definecolor{bleudefrance}{rgb}{0.19, 0.55, 0.91}
\hypersetup{
     colorlinks   = true,
     citecolor    = bleudefrance,
	 linkcolor	  = darkmidnightblue,
	 urlcolor     = darkpastelgreen
}

\title{Private Data Stream Analysis for Universal Symmetric Norm Estimation}
\author{
Vladimir Braverman
\\Rice University\thanks{E-mail: \email{vb21@rice.edu}. Work done in part while at Johns Hopkins University.}
\and
Joel Manning
\\Carnegie Mellon University\thanks{E-mail: \email{joelmanning@cmu.edu}}
\and
Zhiwei Steven Wu
\\Carnegie Mellon University\thanks{E-mail: \email{zhiweiw@andrew.cmu.edu}. ZSW was supported in part by the NSF Award \#2120667 and a Cisco Research Grant.}
\and
Samson Zhou
\\UC Berkeley and Rice University\thanks{E-mail: \email{samsonzhou@gmail.com}. Work done in part while at Carnegie Mellon University.}
}
\date{\today}

\begin{document}

\maketitle

\begin{abstract}
We study how to release summary statistics on a data stream subject to the constraint of differential privacy. In particular, we focus on releasing the family of \emph{symmetric norms}, which are invariant under sign-flips and coordinate-wise permutations on an input data stream and include $L_p$ norms, $k$-support norms, top-$k$ norms, and the box norm as special cases. Although it may be possible to design and analyze a separate mechanism for each symmetric norm, we propose a general parametrizable framework that differentially privately releases a number of sufficient statistics from which the approximation of all symmetric norms can be simultaneously computed. Our framework partitions the coordinates of the underlying frequency vector into different levels based on their magnitude and releases approximate frequencies for the ``heavy'' coordinates in important levels and releases approximate level sizes for the ``light'' coordinates in important levels.  Surprisingly, our mechanism allows for the release of an \emph{arbitrary} number of symmetric norm approximations without any overhead or additional loss in privacy. Moreover, our mechanism permits $(1+\alpha)$-approximation to each of the symmetric norms and can be implemented using sublinear space in the streaming model for many regimes of the accuracy and privacy parameters.
\end{abstract}

\section{Introduction}
The family of $L_p$ norms represent important statistics on an underlying dataset, where the $L_p$ norm\footnote{$L_p$ for $p\in(0,1)$ does not satisfy the triangle inequality and therefore is not a norm, but is still well-defined/well-motivated and can be computed} of an $n$-dimensional frequency vector $x$ is defined as the number of nonzero coordinates of $x$ for $p=0$ and $L_p(x)=\left(x_1^p+\ldots+x_n^p\right)^{1/p}$ for $p>0$. 
Thus, the $L_0$ norm counts the number of distinct elements in the dataset and, e.g., is used to detect denial of service or port scan attacks in network monitoring~\cite{akella2003detecting,estan2003bitmap}, to understand the magnitude of quantities such as search engine queries or internet graph connectivity in data mining~\cite{palmer2001connectivity}, to manage workload in database design~\cite{FinkelsteinST88}, and to select a minimum-cost query plan in query optimization~\cite{SelingerACLP79}. 
The $L_1$ norm computes the total number of elements in the dataset and, e.g., is used for data mining~\cite{CormodeMR05} and hypothesis testing~\cite{IndykM08}, while the $L_2$ norm, e.g., is used for training random forests in machine learning~\cite{breiman2001random}, computing the Gini index in statistics~\cite{lorenz1905methods,gini1912variabilita}, and network anomaly detection in traffic monitoring~\cite{krishnamurthy2003sketch,ThorupZ04}, in particular in the context of heavy-hitters, e.g.,~\cite{CharikarCF04,BravermanCIW16,BravermanCINWW17,BravermanGLWZ18,LiuZRBR20,BlockiLMZ23}. 
More generally, $L_p$ norms for $p\in(0,2)$ have been used for entropy estimation~\cite{HarveyNO08}. 
Consequently, $L_p$ estimation has been extensively studied in the data stream model~\cite{AlonMS99,IndykW05,Indyk06,Li08,KaneNPW11,Andoni17,BravermanVWY18,GangulyW18,WoodruffZ21,WoodruffZ21b}. 
The simplest streaming model is perhaps the insertion-only model, in which a sequence of $m$ updates increments coordinates of an $n$-dimensional frequency vector $x$ and the goal is to compute or approximate some statistic of $x$ in space that is sublinear in both $m$ and $n$. 
For a more formal introduction to the streaming model, see \secref{sec:prelim:stream}.

In many cases, the underlying dataset contains sensitive information that should not be leaked. 
Hence, an active line of work has focused on estimating $L_p$ norms for various values of $p$, while preserving differential privacy~\cite{MirMNW11,BlockiBDS12,Smith0T20,BuGKLST21,WangPS22}. 

\begin{definition}[Differential privacy]
\deflab{def:dp}
\cite{DworkMNS06}
Given $\eps>0$ and $\delta\in(0,1)$, a randomized algorithm $\calA:\frakU^*\to\calY$ is $(\eps,\delta)$-differentially private if, for every neighboring streams $\frakS$ and $\frakS'$ and for all $E\subseteq\calY$,
\[\PPr{\calA(\frakS)\in E}\le e^{\eps}\cdot\PPr{\calA(\frakS')\in E}+\delta.\]
\end{definition}

For example, \cite{BlockiBDS12} showed that the Johnson-Lindenstrauss transformation preserves differential privacy (DP), thereby showing one of the main techniques in the streaming model for $L_2$ estimation already guarantees DP. 
Similarly, \cite{Smith0T20} showed that the Flajolet-Martin sketch, which is one of the main approaches for $L_0$ estimation in the streaming model, also preserves DP. 

However, algorithmic designs for $L_p$ estimation in the streaming model differ greatly and require individual analysis to ensure DP, especially because it is known that for some problems, guaranteeing DP provably requires more space~\cite{DinurSWZ23}. 
Unfortunately, the privacy and utility analysis can be quite difficult due to the complexity of the various techniques.  
This is especially pronounced in the work of \cite{WangPS22}, who studied the $p$-stable sketch~\cite{Indyk06}, which estimates the $L_p$ norm for $p\in(0,2]$. 
\cite{WangPS22} showed that for $p\in(0,1]$, the $p$-stable sketch preserves DP, but was unable to show DP for $p\in(1,2]$, even though the general algorithmic approach remains the same. 
Thus the natural question is whether differential privacy can be guaranteed for an approach that simultaneously estimates the $L_p$ norm in the streaming model, for all $p$. 
More generally, the family of $L_p$ norms are all symmetric norms, which are invariant under sign-flips and coordinate-wise permutations on an input data stream. 
Symmetric norms thus also include other important families of norms such as the $k$-support norms and the top-$k$ norms.

\subsection{Our Contributions} 
In this paper, we show that not only does there exist a differentially private algorithm for the estimation of symmetric norms in the streaming model, but also that there exists an algorithm that privately releases a set of statistics, from which estimates of all (properly parametrized) symmetric norms can be simultaneously computed. 
To illustrate the difference, suppose we wanted to release approximations of the $L_p$ norm of the stream for $k$ different values of $p$. 
To guarantee $(\eps,\delta)$-DP for the set of $k$ statistics, we would need, by advanced composition, to demand $\left(\O{\frac{\eps}{\sqrt{k}}},\O{\frac{\delta}{k}}\right)$-DP from $k$ instances of a single differentially private $L_p$-estimation algorithm, corresponding to the $k$ different values of $p$. 
Due to accuracy-privacy tradeoffs, the quality of the estimation will degrade severely as $k$ increases. 
For an extreme example, consider when $k$ is some large polynomial of $n$ and $m$ so that the added noise will also be polynomial in $n$ and $m$, and then there is no utility at all -- the private algorithm might as well just release $0$ for all queries! 

In contrast, our algorithm releases a single set $C$ of private statistics.
By post-processing, we can then estimate the $L_p$ norms for $k$ different values of $p$ while only requiring $(\eps,\delta)$-DP from $C$. 
Hence, our algorithm can simultaneously handle any large number of estimations of symmetric norms without compromising the quality of approximation. 

We first informally introduce the definition of the maximum modulus of concentration of a norm, which measures the worst-case ratio of the maximum value of a norm on the $L_2$-unit sphere to the median value of a norm on the $L_2$-unit sphere, where the median can be taken over any restriction of the coordinates. 
Intuitively, maximum modulus of concentration of a norm quantifies the complexity of computing a norm. 
For example, the $L_1$ norm is generally ``easy'' to compute and has maximum modulus of concentration $\O{\log n}$. 
See \defref{def:mmc} for a more formal definition. 
Then our main result can informally be stated as follows: 

\begin{theorem}[Informal]
\thmlab{thm:main}
There exists a $(\eps,\delta)$-differentially private algorithm that outputs a set $C$, from which the $(1+\alpha)$-approximation to any norm, with \emph{maximum modulus of concentration} at most $M$ of a vector $x\in\mathbb{R}^n$ induced by a stream of length $\poly(n)$ can be computed, with probability at least $1-\delta$. 
The algorithm uses $M^2\cdot\poly\left(\frac{1}{\alpha},\frac{1}{\eps},\log n,\log\frac{1}{\delta}\right)$ bits of space. 
\end{theorem}

We remark that as is standard in differential privacy on data streams, both the privacy parameter $\eps$ and the accuracy parameter $\alpha$ cannot be too small or the additive noise will be too large and cannot be absorbed into the $(1+\alpha)$-multiplicative bounds. 
See \thmref{thm:main:formal} for the formal statement of \thmref{thm:main} describing these bounds.  

We also remark that in the statement of \thmref{thm:main}, the $\delta$ failure parameter of approximate DP is equal to the failure parameter $\delta$ of the utility guarantees of the algorithm. 
More generally, if the desired failure probability $\delta'$ of the utility guarantee is not equal to the privacy parameter $\delta$, then the dependencies will change from $\log\frac{1}{\delta}$ to $\log\frac{1}{\delta\delta'}$. 

We emphasize that prior to our work, there is no algorithm that can handle private symmetric norm estimation for arbitrary symmetric norms, much less simultaneously for all parametrized symmetric norms. 
Although there is specific analysis for various norm estimation algorithms, e.g., see the discussion on related work in \secref{sec:related}, these algorithms require a specific predetermined norm for their input. 
Thus a separate private algorithm must be run for each estimation, which increases the overall space. 
Moreover, for a large number of queries, the privacy parameter will need to be much smaller due to the composition of privacy, and thus to ensure privacy, the utility of each algorithm is provably poor. 
Our algorithm sidesteps both the space and accuracy problems and is the first and only work to do so, as of yet. 

\paragraph{Applications.}
We briefly describe a number of specific symmetric norms that are handled by \thmref{thm:main} and commonly used across various applications in machine learning. 
We first note the following parameterization of the previously discussed $L_p$ norms. 
\begin{lemma}
\lemlab{lem:mmc:lp}
\cite{milman2009asymptotic,klartag2007small}
For $L_p$ norms, we have that $\mmc(L)=\O{\log n}$ for $p\in[1,2]$ and $\mmc(L)=\O{n^{1/2-1/p}}$ for $p>2$. 
\end{lemma}
Thus our algorithm immediately introduces a differentially private mechanism for the approximation of $L_p$ norms that unlike previous work, e.g.,~\cite{BlockiBDS12,Sheffet19,ChoiDKY20,Smith0T20,BuGKLST21,WangPS22}, does not need to provide separate analysis for specific values of $p$. 
Moreover for constant-factor approximation, the space complexity is tight with the \emph{optimal} $L_p$-approximation algorithms that do not consider privacy, up to polylogarithmic factors~\cite{KaneNW10,LiW13,Ganguly15,WoodruffZ21} in the universe size $n$.  

\begin{definition}[$Q$-norm and $Q'$-norm]
We call a norm $L$ a $Q$-norm if there exists a symmetric norm $L'$ such that $L(x)=L'(x^2)^{1/2}$ for all $x\in\mathbb{R}^n$. 
Here, we use $x^2$ to denote the coordinate-wise square power of $x$. 
We also call a norm $L'$ a $Q'$-norm if its dual norm is a $Q$-norm. 
\end{definition}
The family of $Q'$-norms includes the $L_p$ norms for $1\le p\le 2$, the $k$-support norm, and the box norm~\cite{bhatia2013matrix} and thus $Q'$-norms have been proposed to regularize sparse recovery problems in machine learning. 
For instance, \cite{ArgyriouFS12} showed that $Q'$ norms have tighter relaxations than elastic nets and can thus be more effective for sparse prediction. 
Similarly, \cite{McDonaldPS14} used $Q'$ norms to optimize sparse prediction algorithms for multitask clustering.  
\begin{lemma}
\lemlab{lem:mmc:qnorm}
\cite{BlasiokBCKY17}
$\mmc(L)=\O{\log n}$ for every $Q'$-norm $L$. 
\end{lemma}
\thmref{thm:main} and \lemref{lem:mmc:qnorm} thus present a differentially private algorithm for $Q'$-norm approximation that uses polylogarithmic space.  
\begin{definition}[Top-$k$ norm]
The top-$k$ norm for a vector $x\in\mathbb{R}^n$ is the sum of the largest $k$ coordinates of $|x|$, where we use $|x|$ to denote the vector whose entries are the coordinate-wise absolute value of $x$. 
\end{definition}
The top-$k$ norm is frequently used to understand the more general Ky Fan $k$-norm~\cite{WuDST14}, which is used to regularize optimization problems in numerical linear algebra. 
Whereas the Ky Fan $k$ norm is defined as the sum of the $k$ largest singular values of a matrix, the top-$k$ norm is equivalent to the Ky Fan $k$ norm when the input vector $x$ represents the vector of the singular values of the matrix. 
\begin{lemma}
\lemlab{lem:mmc:topk}
\cite{BlasiokBCKY17}
$\mmc(L)=\tO{\sqrt{\frac{n}{k}}}$ for the top-$k$ norm $L$. 
\end{lemma}
In particular, the top-$k$ norm for a vector of singular values when $k=n$ is equivalent to the Schatten-$1$ norm of a matrix, which is a common metric for matrix fitting problems such as low-rank approximation~\cite{LiW20}. 

\begin{definition}[Shannon entropy]
For a frequency vector $v\in\mathbb{R}^n$, we define the Shannon entropy by $H(v)=-\sum_{i=1}^n v_i\log v_i$. 
\end{definition}

To achieve an additive approximation to the Shannon entropy, we instead compute a multiplicative approximation to the exponential form, as follows:

\begin{observation}
\obslab{obs:entropy:addmult}
A $(1+\alpha)$-multiplicative approximation of the function $h(v):=2^{H(v)}$ corresponds to an $\alpha$-additive approximation of the Shannon Entropy $H(v)$ (and vice versa). 
\end{observation}

Moreover, computing a $(1+\alpha)$-approximation to $2^{H(v)}$ can be achieved through computing a $(1+\alpha)$-approximation to various $L_p$ norms for $p\in(0,2)$.

\begin{lemma}[Section 3.3 in~\cite{HarveyNO08}]
\lemlab{lem:entropy:reduction}
Let $k=\log\frac{1}{\alpha}+\log\log m$ and $\alpha'=\frac{\alpha}{12(k+1)^3\log m}$. 
There exists an explicit set $\{y_0,\ldots,y_k\}$ with $y_i\in(0,2)$ for all $i$ and a post-processing function that takes $(1+\alpha')$-approximations to $F_{y_i}(x)$, i.e., the $(y_i)$-th frequency moment of $x$, and outputs a $(1+\alpha)$-approximation to $h(v)=2^{H(x)}$. 
Furthermore, the set $\{y_0,\ldots,y_k\}$ and post-processing function are both efficiently computable, i.e., polynomial runtime. 
\end{lemma}

Since our mechanism releases a private set of statistics from which $(1+\alpha)$-approximations to $L_p$ norms can be computed for any $p\in(0,2)$, then our mechanism also privately achieves an additive $\alpha$-approximation to Shannon entropy. 

\subsection{Algorithmic Intuition and Overview}
Our starting point is the $L_p$ estimation algorithm of \cite{IndykW05}, which was parametrized by \cite{BlasiokBCKY17} to handle symmetric norms. 
For a $(1+\alpha)$-approximation, the algorithm partitions the $n$ coordinates of the frequency vector $x$ into powers of $\xi$-based on their magnitudes, where $\xi>1$ is a fixed function of $\alpha$. 
Each partition forms a level set, so that the $i$-th level set consists of the coordinates of $x$ with frequency $[\xi^i,\xi^{i+1})$, but \cite{IndykW05,BlasiokBCKY17} showed that it suffices to accurately count the size of each \emph{important} level set and zero out to the other level sets, where a level set is considered important if its size is large enough to contribute an $\frac{\alpha^2}{\log m}$ fraction of the symmetric norm. 
In other words, if $\tilde{x}$ is a vector whose coordinates match those of $x$ in important levels sets and are $0$ elsewhere, then $(1-\alpha)L(x)\le L(\tilde{x})\le(1+\alpha)L(x)$. 
We formalize the definition of importance in \secref{sec:prelims:sym}. 

\paragraph{Private symmetric norm estimation in the centralized setting.} 
To preserve $(\eps,\delta)$-differential privacy, one initial approach would be to view the frequency vector as a histogram and add Laplacian noise with scale $\O{\frac{1}{\eps}}$ to the frequency of each element. 
However, the level sets consisting of elements with frequencies between $[\xi^i,\xi^{i+1})$ for small $i$, say $i=0$, could be largely perturbed by such Laplacian noise. 
For example, it is possible that for some coordinate $j$ in an important level set, we have $x_j=1$, in which case adding Laplacian noise with scale $\O{\frac{1}{\eps}}$ to $x_j$ will heavily distort the coordinate. 
This can happen to all coordinates in the important level set, which results in an inaccurate estimation of the norm. 

Fortunately, if $i$ is small, the corresponding level set must contain a large number of elements if it is important, so it seems possible to privately release the size $\Gamma_i$ of the level set. 
Indeed, we can show that the $L_1$ sensitivity of the vector corresponding to level set sizes is small and so we can add Laplacian noise with scale $\O{\frac{1}{\eps}}$ to each level set size. 
Hence if the level set has size $\Gamma_i$ roughly $\Omega\left(\frac{1}{\alpha\eps}\right)$, then the Laplacian noise will affect $\Gamma_i$ by a $(1+\alpha)$-factor. 

Unfortunately, there can be level sets that are both important and small in size. 
For example, if there is a single element with frequency $m$, then the size of the corresponding level set is just one. 
Then adding Laplacian noise with scale $\O{\frac{1}{\eps}}$ will severely affect the size of the level set and thus the estimation of the symmetric norm. 
On the other hand, for $m>\frac{1}{\alpha\eps}$, the frequency of the coordinate is quite large so again it seems like we can just add Laplacian noise with scale $\O{\frac{1}{\eps}}$ and output the noisy frequency of the coordinate. 

\paragraph{New approach: classifying and separately handling high, medium, and low frequency levels.} 
The main takeaway from these challenges is that we should handle different level sets separately. 
For the level sets of small coordinates, the important level sets must have large size and thus we would like to release noisy sizes. 
For the important level sets of large coordinates, we would like to release noisy frequencies of the coordinates. 

In that vein, we partition the levels into three groups after defining thresholds $T_1$ and $T_2$, with $T_1>T_2$. 
We define the ``high frequency levels'' as the levels whose coordinates exceed $T_1$ in frequency. 
The intuition is that because the high frequency levels have such large magnitude, their frequencies can be well-approximated by running an $L_2$-heavy hitters algorithm on the stream $S$. 

We define the ``medium frequency levels'' as the levels whose coordinates are between $T_1$ and $T_2$ in frequency. 
These coordinates are not large enough to be detected by running an $L_2$-heavy hitters algorithm on the stream $S$. 
However, the sizes of these level sets must be large if the level set is important. 
Thus there exists a substream $S_j$ for which a large number of these coordinates are subsampled and their frequencies can be well-approximated by running an $L_2$-heavy hitters algorithm on the substream $S_j$. 

Finally, we define the ``low frequency levels'' as the levels whose coordinates are less than $T_2$ in frequency. 
These coordinates are small enough that we cannot add Laplacian noise to their frequencies without affecting the level sets they are mapped to. 
Instead, we show that the $L_1$ sensitivity for the level set estimations is particularly small for the low frequency levels. 
Thus, for these frequency levels, we report the size of the frequency levels rather than the approximate frequencies of the heavy-hitters. 
We remark that if our goal was to just approximate the symmetric norms without preserving differential privacy, then it would suffice to just consider the high and medium frequency levels, since the low frequency levels are particularly problematic when Laplacian noise is added to the frequency vector. 
We also remark that we only use the thresholds $T_1$ and $T_2$ for the purposes of describing our algorithm -- in the actual implementation of the algorithm, the thresholds $T_1$ and $T_2$ will be implicitly defined by each of the substreams. 
We summarize our new approach in \figref{fig:flowchart}.  

\begin{figure*}[!htb]
\centering
\begin{tikzpicture}[scale=0.95]
\draw (0,0) rectangle+(3.5,1);
\node at (1.75,0.5){Low Level Sets};

\draw[purple,->] (3.6,0.5) -- (3.9,0.5);

\draw (4,0) rectangle+(3.5,1);
\node at (1.75+4,0.5){Subsampling};

\draw[purple,->] (7.6,0.5) -- (7.9,0.5);

\draw (8,0) rectangle+(4.5,1);
\node at (8+2.25,0.5){Noisy level set sizes};

\node at (14.75,0.5){Private level set sizes};

\draw (0,1.5) rectangle+(3.5,1);
\node at (1.75,2){Medium Level Sets};

\draw[purple,->] (3.6,2) -- (3.9,2);

\draw (4,1.5) rectangle+(3.5,1);
\node at (1.75+4,2){Subsampling};

\draw[purple,->] (7.6,2) -- (7.9,2);

\draw (8,1.5) rectangle+(4.5,1);
\node at (8+2.25,2){$\privcountsketch$};

\node at (14.75,2){Private level set sizes};

\draw (0,3) rectangle+(3.5,1);
\node at (1.75,3.5){High Level Sets};

\draw[purple,->] (3.6,3.5) -- (4.9,3.5);

\draw (5,3) rectangle+(4.5,1);
\node at (5+2.25,3.5){$\privcountsketch$};

\node at (13.75,3.5){Private coordinate magnitudes};

\end{tikzpicture}
\caption{Illustration of separate handling of the high, medium, and low level sets.}
\figlab{fig:flowchart}
\end{figure*}
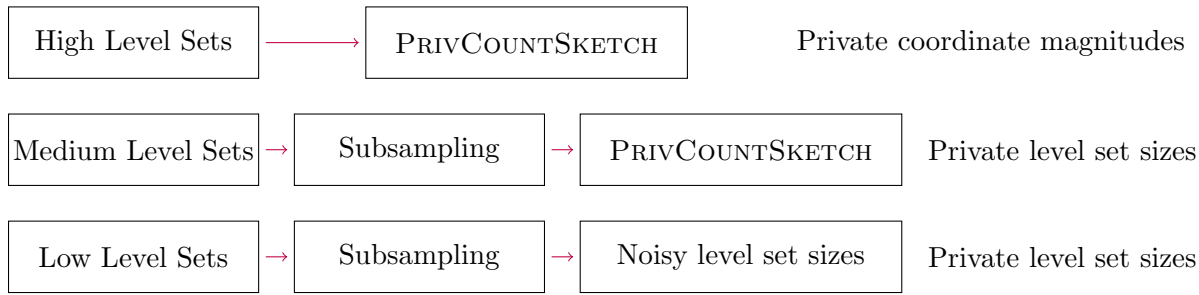

\paragraph{Private symmetric norm estimation in the streaming model.} 
Although the previously discussed intuition builds towards a working algorithm, the main caveat is that so far, we have mainly discussed the centralized model, where space is not restricted and so each coordinate and thus each level set size can be counted exactly. 
In the streaming model, we cannot explicitly track the frequency vector, or even the frequencies of a constant fraction of coordinates. 
Instead, to estimate the sizes of each level set, \cite{IndykW05,BlasiokBCKY17} take the stream $S$ and form $s=\O{\log n}$ substreams $S_1,\ldots,S_s$, where the $j$-th substream is created by sampling the universe of size $n$ at a rate of $\frac{1}{2^{j-1}}$. 
Then $S_j$ will only consist of the stream updates to the particular coordinates of $x$ that are sampled. 
Thus in expectation, the frequency vector induced by $S_j$ will have sparsity $\frac{\|x\|_0}{2^{j-1}}$. 
Similarly, if a level set $i$ has size $\Gamma_i$, then $\frac{\Gamma_i}{2^{j-1}}$ of its members will be sampled in $S_j$ in expectation. 
It can then be shown through a variance argument that if level set $i$ is important, then there exists an explicit substream $j$ from which $\Gamma_i$ can be well-approximated using the $L_2$-heavy hitter algorithm $\countsketch$ and as a result, the symmetric norm of $x$ can be well-approximated. 
The main point of the subsampling approach is that if there exists a level set with large size consisting of small coordinates, then the coordinates will not be detected by the $\countsketch$ on $S$, but because $S_j$ has significantly smaller $L_2$ norm, then the coordinates will be detected by $\countsketch$ on $S_j$. 

However, adapting the subsampling and heavy-hitter approach introduces additional challenges for privacy. 
For instance, we can analyze the $L_2$-heavy hitter algorithm $\countsketch$ and show that although the $L_1$ sensitivity of the estimated frequency for a single coordinate is small, the $L_1$ sensitivity of the estimated frequency vector for all the coordinates may be large. 
Instead, we use the view that $\countsketch$ is a composition function that first only estimates frequencies for the top $\poly\left(\frac{1}{\alpha},\frac{1}{\eps},\log n\right)$ and then outputs only those estimates that are above a certain threshold. 
Similarly, the Laplacian noise added to privately use $\countsketch$ can alter the sizes of a significant number of level sets for small coordinates. 
Thus for the small coordinates (corresponding to the substreams $S_j$ with large $j$), we invoke $\countsketch$ with much higher accuracy, so that with high probability, it will return \emph{exactly} the frequencies for the small coordinates. 
For example, note that if the frequency $x_k$ of a coordinate $k\in[n]$ is at most $\frac{1}{2\alpha^2\eps}$, then any $(1+\alpha^2\eps)$-approximation to $x_k$ can be rounded to exactly recover $x_k$. 
This decreases the $L_1$ sensitivity of the vector of estimated level set sizes, therefore allowing us to add Laplacian noise without greatly affecting the quality of approximation. 

\subsection{Related Work}
\seclab{sec:related}
Non-private $L_p$ norm estimation is one of the fundamental problems in the streaming model, beginning with \cite{AlonMS99}'s seminal work that tracks the inner product of the frequency vector with a random sign vector for $L_2$ estimation (as well as a telescoping argument for integer $p>0$). 
\cite{Indyk06,Li08} later showed that this approach could be generalized for $p\in(0,2]$ by tracking the inner product of the frequency vector with a vector with randomly generated $p$-stable variables, which only exist for $p\in(0,2]$. 
For $p>2$, \cite{Andoni17} gave an $L_p$ estimation algorithm using the max-stability property of exponential random variables. 
More generally, \cite{IndykW05} introduced the framework of subsampling and using heavy-hitters for $L_p$ estimation, which \cite{BlasiokBCKY17} parametrized to all symmetric norms. 
It should be emphasized that these techniques all handle the more general turnstile model, in which $\pm1$ updates are allowed to each coordinate, rather than single positive increments. 
Hence our techniques also extend to the turnstile model with a minor change on the conditions. 

More recently, \cite{BlockiGMZ22,Tetek22} given a general framework for converting non-private approximation algorithms into private approximation algorithms, provided that the accuracy of these algorithms could be tuned with an input parameter $\eps>0$, i.e., the algorithms can achieve $(1+\eps)$-approximation for a wide range of $\eps>0$. 
Their results presented a solution that addresses the difficulty of adapting privacy specifically to each non-private algorithm separately. 
However, their framework only applies to problems with scalar outputs and thus do not handle synthetic data release. 
Therefore, privately answering multiple norm queries while circumventing composition bounds is still a challenge that their results cannot handle. 

Symmetric norms have also recently received attention in other big data models as well. 
\cite{AndoniNNRW17} studied approximate near neighbors for general symmetric norms while \cite{LiuMVSB16} studied symmetric norm estimation for network monitoring. 
\cite{SongWYZZ19} considered Orlicz norm regression and other loss functions where the penalty is a symmetric norm. 
\cite{BravermanWZ21} gave an algorithm to approximate the symmetric norm in the sliding window model, where updates in the data stream implicitly expire after a fixed amount of time. 

Specific cases of private $L_p$ estimation in the streaming model have also been previously well-studied. 
\cite{ChoiDKY20,Smith0T20} studied private $L_0$ estimation using the Flajolet-Martin sketch, while \cite{WangPS22} studied private $L_p$ estimation for $p\in(0,1]$ using the $p$-stable sketch and \cite{BlockiBDS12,Sheffet19,ChoiDKY20,BuGKLST21} studied private $L_2$ estimation using the Johnson-Lindenstrauss projection. 
Specifically, \cite{BlockiBDS12} gave an $(\eps,\delta)$-DP algorithm for $L_2$ estimation that achieves a $(1+\eps)$-approximation while using $\O{\frac{1}{\eps^2}\log n\log\frac{1}{\delta}}$ bits of space and \cite{WangPS22} gave an $(\eps,\delta)$-DP algorithm for $L_p$ estimation that achieves a $(1+\alpha)$-approximation while using $\O{\frac{1}{\alpha^2}\log n\log\frac{1}{\delta}}$ bits of space for constant $\eps$ and $p\in(0,1)$. 
For fractional $p>1$, private distribution estimation algorithms~\cite{AcsCC12,XuZXYYW13,BassilyS15,WangLLSL20} can be used to approximate the $L_p$ norm, but since the algorithms provide information over a much larger distribution, e.g., much larger histograms of frequencies, the privacy-accuracy trade-off is sub-optimal and the space complexity is exponentially worse. 

The related problem of privately releasing heavy-hitters in big data models has also been well-studied. 
\cite{ChanLSX12} studied the problem of continually releasing $L_1$-heavy hitters in a stream, while \cite{DworkNPRY10} studied $L_1$-heavy hitters and other problems in the pan-private streaming model. 
The heavy-hitter problem has also received significant attention in the local model, e.g.,~\cite{BassilyS15,DingKY17,AcharyaS19,BunNS19,BassilyNST20}, where individual users should locally randomize their data before sending differentially private information to an untrusted server that aggregates the statistics across all users. 

\section{Preliminaries}
In this section, we introduce definitions and simple or well-known results from differential privacy, sketching algorithms, and symmetric norms. 
For notation, we use $[n]$ for an integer $n>0$ to denote the set $\{1,\ldots,n\}$. 
We also use the notation $\poly(n)$ to represent a constant degree polynomial in $n$ and we say an event occurs \emph{with high probability} if the event holds with probability $1-\frac{1}{\poly(n)}$. 
Similarly, we use $\polylog(n)$ to denote $\poly(\log n)$. 
Given a vector $x\in\mathbb{R}^n$, we define its second frequency moment $F_2(x)=x_1^2+\ldots+x_n^2$. 
Finally, for a parameter $c\ge 1$, we say that $X$ provides a $C$-approximation to a quantity $Y$ if $\frac{X}{C}\le Y\le C\cdot X$.

\subsection{Differential privacy}
We first recall standard definitions and results from differential privacy. 
We define the concept of neighboring streams used in \defref{def:dp}. 
\begin{definition}[Neighboring streams]
Data streams $\frakS$ and $\frakS'$ are neighboring if there exists a single update $i\in[m]$ such that $u_i\neq u'_i$, where $u_1,\ldots,u_m$ are the updates of $\frakS$ and $u'_1,\ldots,u'_m$ are the updates of $\frakS'$. 
\end{definition}
A standard approach to guarantee differential privacy is to add Laplacian noise, which is drawn from the following distribution. 
\begin{definition}[Laplace distribution]
A random variable $x$ is drawn from the \emph{Laplace distribution} with mean $\mu$ and scale $s>0$, $x\sim\Lap(\mu,s)$, if the probability density function of $x$ is $\frac{1}{2s}\exp\left(-\frac{|x-\mu|}{s}\right)$. 
If $\mu=0$, we say $x\sim\Lap(s)$. 
\end{definition}
The amount of Laplacian noise that must be added to an output depends on the $L_1$ sensitivity of the function. 
\begin{definition}[$L_1$ sensitivity]
We define the $L_1$ sensitivity of a function $f:\mathbb{N}^{|\frakU|}\to\mathbb{R}^k$ by
\[\Delta_f=\max_{x,y\in\mathbb{N}^{|\frakU|},\|x-y\|_1=1}\|f(x)-f(y)\|_1.\]
\end{definition}
Intuitively, the $L_1$ sensitivity of a function is the largest amount that $f$ can change when a single update in the stream that defines $f$ changes.
\begin{definition}[Laplace mechanism]
Given a function $f:\mathbb{N}^{|\frakU|}\to\mathbb{R}^k$, we define the Laplace mechanism by
\[\calM_L(x,f,\eps)=f(x)+(X_1,\ldots,X_k),\]
where $X_i\sim\Lap(\Delta f/\eps)$.
\end{definition}
The Laplace mechanism is one of the most fundamental ways to ensure differential privacy:
\begin{theorem}
\cite{DworkR14}
The Laplace mechanism preserves $(\eps,0)$-differential privacy.
\end{theorem}
Privately releasing multiple statistics that are individually differentially private can also be done, but comes at a slight cost. 
\begin{theorem}[Composition and post-processing of differential privacy]
\cite{DworkR14}
\thmlab{thm:dp:comp}
Let $\calA_i:\frakU_i\to X_i$ be an $(\eps_i,\delta_i)$-differential private algorithm for $i\in[k]$. 
Then $\calA_{[k]}(x)=(\calA_1(x),\ldots,\calA_k(x))$ is $\left(\sum_{i=1}^k\eps_i,\sum_{i=1}^k\delta_i\right)$-differentially private. 
Furthermore, if $g_i:X_i\to X'_i$ is an arbitrary random mapping, then $g_i(\calM_i(x))$ is $(\eps_i,\delta_i)$-differentially private. 
\end{theorem}
Although there exists more sophisticated approaches for composition, such as advanced composition, we do not need them for our purposes. 

\subsection{Streaming and Sketching Algorithms}
\seclab{sec:prelim:stream}
In the streaming model, a frequency vector $x\in\mathbb{R}^n$ is induced by a sequence of updates. 
In the insertion-only streaming model, $x$ is defined through a stream of $m$ updates $u_1,\ldots,u_m$, where $u_t\in[n]$ for each $t\in[m]$ so that $x_i=|\{t\in[m]\,|\,u_t=i\}|$ for all $i\in[n]$. 
In other words, $x_i$ is the number of times that $i\in[n]$ appears in the stream. 
We remark that our techniques generalize to some degree to turnstile streams, where each update is an ordered pair $u_t=(\Delta_t,c_t)$, so that the $t$-th update changes the $c_t$-th coordinate by $\Delta_t$, i.e., $c_t\in[n]$ is a coordinate and $\Delta_t\in[-M,M]$ for some parameter $M>0$. 
In this turnstile model, the vector $x$ is defined so that $x_i=\sum_{t:c_t=i}\Delta_t$ for all $i\in[n]$. 
Although our techniques can apply to the general turnstile model with a minor change on the conditions and assumptions, we shall work with the insertion-only streaming model throughout the remainder of the paper. 

Given a frequency vector $x\in\mathbb{R}^n$ on a data stream, the $\ams$ algorithm for $L_2$-estimation first generates a sign vector $\sigma\in\{-1,+1\}^n$ and sets $S_1=(\langle\sigma,x\rangle)^2$. 
We remark that to maintain $\sigma$ in small space, it suffices for the coordinates of the sign vector $\sigma$ to be $4$-wise independent and therefore it suffices to randomly generate and store a $4$-wise independent hash function. 
The $\ams$ algorithm then repeats this process $b=\frac{6}{\alpha^2}$ independent times to obtain dot products $S_1,\ldots,S_b$, sets $Z^2$ to be the arithmetic mean of $S_1,\ldots,S_b$, and reports $Z$.

We define the $L_2$ norm of a vector $x\in\mathbb{R}^n$ by $L_2(x)=\sqrt{x_1^2+\ldots+x_n^2}$. 


\begin{definition}[$\nu$-approximate $\eta$ $L_2$-heavy hitters problem]
Given an accuracy parameter $\nu\in(0,1)$, a threshold parameter $\eta$, and a frequency vector $x\in\mathbb{R}^n$, compute a set $H\subseteq[n]$ and a set of approximations $\widehat{x_k}$ for all $k\in H$ such that: 
\begin{enumerate}
\item
If $x_k\ge\eta L_2(x)$ for any $k\in[n]$, then $k\in H$, so that $H$ contains all $\eta$ $L_2$-heavy hitters of $x$. 
\item
There exists a universal constant $C\in(0,1)$ so that if $x_k\le\frac{C\eta}{2} L_2(x)$ for any $k\in[n]$, then $k\notin H$, so that $H$ does not contain any index that is not an $\frac{C\eta}{2}$ $L_2$-heavy hitter of $x$. 
\item
If $k\in H$ for any $k\in[n]$, then compute $(1\pm\nu)$-approximation to the frequency $x_k$, i.e., a value $\widehat{x_k}$ such that $(1-\nu)x_k\le\widehat{x_k}\le(1+\nu)x_k$. 
\end{enumerate}
\end{definition}

The well-known $\countsketch$ algorithm can be parametrized to provide an estimated frequency to each item and then releases the approximate frequencies of each item that surpasses a threshold proportional to the output of $\ams$: 
\begin{theorem}[CountSketch for $\nu$-approximate $\eta$ $L_2$-heavy hitters]
\cite{CharikarCF04}
\thmlab{thm:countsketch} 
There exists a one-pass streaming algorithm $\countsketch$ that takes an accuracy parameter $\nu\in(0,1)$ and a threshold parameter $\eta^2$ and outputs a list $H$ that contains all indices $k\in[n]$ of an underlying frequency vector $x$ with $x_k\ge\eta\,L_2(x)$ and no index $k\in[n]$ with $x_k\le\eta(1-\nu)\,L_2(x)$. 
For each $k\in H$, $\countsketch$ also reports a estimated frequency $\widehat{x_k}$ such that $(1-\nu)x_k\le\widehat{x_k}\le(1+\nu)x_k$. 
The algorithm uses $\O{\frac{1}{\eta^2\nu^2}\log^2 n}$ bits of space and succeeds with probability $1-\frac{1}{\poly(m)}$. 
\end{theorem}

\begin{algorithm}[!htb]
\caption{Heavy-hitter algorithm $\countsketch$}
\alglab{alg:cs}
\begin{algorithmic}[1]
\Require{Stream $\frakS$ inducing frequency vector $x\in\mathbb{R}^n$, accuracy parameter $\nu\in(0,1)$, and threshold parameter $\eta\in(0,1)$}
\Ensure{$L_2$ Heavy-hitter algorithm}
\State{$r\gets\O{\log n}$, $b\gets\O{\frac{1}{\eta^2\nu^2}}$}
\State{Pick hash functions $h^{(1)},\ldots,h^{(r)}:[n]\to[b]$ and $s^{(1)},\ldots,s^{(r)}:[n]\to\{-1,+1\}$}
\State{$S_{i,j}\gets0$ for $(i,j)\in[r]\times[b]$}
\For{each update $u_i\in[n]$, $i\in[m]$}
\For{each $j\in[r]$}
\State{$b_{i,j}\gets h^{(j)}(u_i)$ and $s_{i,j}\gets s^{(j)}(u_i)$}
\State{$S_{j,b_{i,j}}\gets S_{j,b_{i,j}}+s_{i,j}$}
\EndFor
\EndFor
\For{each $i\in[n]$}
\State{$b_{i,j}\gets h^{(j)}(u_i)$ for each $j\in[r]$}
\State{\Return $\median_{j\in[r]}|S_{j,b_{i,j}}|$ as the estimated frequency for $x_i$}
\EndFor
\end{algorithmic}
\end{algorithm}
We provide the standard analysis of the $L_1$ sensitivity of $\countsketch$ through the following statements.
\begin{lemma}
\lemlab{lem:sens:ams:one}
Let $s\in\{-1,+1\}^n$ and $x,x'\in\mathbb{R}^n$ with $\max(\|x-x'\|_0,\|x-x'\|_1)\le 2$. 
Then $|\langle s,x\rangle-\langle s,x'\rangle|\le 2$.
\end{lemma}
\begin{proof}
Note that we have $\langle s,x\rangle-\langle s,x'\rangle=\langle s,x-x'\rangle$. 
Since $\max(\|x-x'\|_0,\|x-x'\|_1)\le 2$, then $|\langle s,x-x'\rangle|\le2\|s\|_\infty=2$. 
\end{proof}


\begin{lemma}[Sensitivity of median]
\lemlab{lem:sens:ams:three}
Let $x,y\in\mathbb{R}^n$ so that $\|x-y\|_\infty\le C$ for some constant $C\ge 0$. 
Then 
\[\left|\median(x_1,\ldots,x_n)-\median(y_1,\ldots,y_n)\right|\le C.\]
\end{lemma}
\begin{proof}
Without loss of generality, suppose $x_1\le\ldots\le x_n$ and first suppose that $n$ is even. 
Then there exist at least $\frac{n}{2}$ indices $i$ such that $x_i\le\median(x_1,\ldots,x_n)$ and thus at least $\frac{n}{2}$ indices $i$ with $y_i\le\median(x_1,\ldots,x_n)+C$. 
Hence, $\median(y_1,\ldots,y_n)\le\median(x_1,\ldots,x_n)+C$. 
Similarly, there exist at least $\frac{n}{2}$ indices $i$ such that $y_i\ge\median(x_1,\ldots,x_n)-C$ and thus, $\median(y_1,\ldots,y_n)\ge\median(x_1,\ldots,x_n)-C$. 
Therefore,
\[\left|\median(x_1,\ldots,x_n)-\median(y_1,\ldots,y_n)\right|\le C,\]
for even $n$.  

For odd $n$, there exist at least $\frac{n+1}{2}$ indices $i$ with $x_i\le\median(x_1,\ldots,x_n)$ and at least $\frac{n+1}{2}$ indices $i$ with $x_i\ge\median(x_1,\ldots,x_n)$, and so there are at least $\frac{n+1}{2}$ indices $i$ with $y_i\le\median(x_1,\ldots,x_n)+C$ and we have $\median(y_1,\ldots,y_n)\le\median(x_1,\ldots,x_n)+C$. 
Similarly, there exist at least $\frac{n+1}{2}$ indices $i$ such that $y_i\ge\median(x_1,\ldots,x_n)-C$, so that
\[\left|\median(x_1,\ldots,x_n)-\median(y_1,\ldots,y_n)\right|\le C,\]
for odd $n$ as well. 
\end{proof}


\begin{lemma}[Sensitivity of CountSketch]
\lemlab{lem:sens:countsketch}
Let $x,x'\in\mathbb{R}^n$ with $\max(\|x-x'\|_0,\|x-x'\|_1)\le 2$. 
Then we have $|\widehat{x_k}-\widehat{x'_k}|\le 2$ for each estimate of the frequency $x_k,x'_k$ of the $k$-th coordinate, $k\in[n]$ by $\countsketch$ with fixed internal randomness.
\end{lemma}
\begin{proof}
Let $b=\O{\frac{1}{\alpha^2}}$ with a sufficiently large constant. 
Then $\countsketch$ first generates $r=\O{\log n}$ hash functions $h^{(1)},\ldots,h^{(r)}:[n]\to[b]$. 
Fix $i\in[r]$ and $j\in[b]$ and consider the subset $S_{i,j}\subseteq[n]$ defined by $S_{i,j}:=\{x\in[n]\,:\, h^{(i)}(x)=j\}$. 
Define $g$ and $g'$ as the restriction of the vectors $x$ and $x'$ to the coordinates of $S_{i,j}$ and $s^{(i,j)}$ as the sign vector $s^{(i)}$ restricted to the coordinates of $S_{i,j}$. 

For each $i\in[r]$ and $k\in[n]$, let $\widehat{x^{(i)}_k}$ be the estimate of $x_k$ by $h^{(i)}$ and $\widehat{\overline{x}^{(i)}_k}$ be the estimate of $x'_k$ by $h^{(i)}$. 
Since $\max(\|g-g'\|_0,\|g-g'\|_1)\le 2$ and $s^{(i,j)}$ is a sign vector, then by \lemref{lem:sens:ams:one}, 
\[|\langle s^{(i,j)},g\rangle-\langle s^{(i,j)},g'\rangle|\le 2\]
and thus
\[|\widehat{x^{(i)}_k}-\widehat{\overline{x}^{(i)}_k}|\le 2\]
for all $i\in[r]$. 
By \lemref{lem:sens:ams:three}, 
\[\left|\median(\widehat{x^{(1)}_k},\ldots,\widehat{x^{(r)}_k})-\median(\widehat{\overline{x}^{(1)}_k},\ldots,\widehat{\overline{x}^{(r)}_k})\right|\le 2.\]
Because $\countsketch$ outputs the median of the estimated frequencies for each coordinates $k\in[n]$, then it follows that $|\widehat{x_k}-\widehat{x'_k}|\le 2$ for the estimates $\widehat{x_k},\widehat{x'_k}$ output by $\countsketch$ with fixed internal randomness on the vectors $x$ and $x'$, for all $k\in[n]$. 
\end{proof}
We can thus use a private variant $\privcountsketch$ of $\countsketch$ by adding noise to each coordinate and then using a standard private threshold routine to ensure differential privacy. 
Specifically, $\privcountsketch$ first uses the $\countsketch$ data structure to obtain an estimated frequency for each coordinate. 
It then adds Laplacian noise with scale parameter $\O{\frac{1}{\eta^2\nu^2}}$ to each estimated frequency, since by \lemref{lem:sens:countsketch}, the sensitivity of the estimated frequency for each coordinate by $\countsketch$ is at most $2$. 
It then acquires a threshold $T$ from a private $L_2$ norm estimation algorithm, e.g., $\ams$ with Laplacian noise, and releases all coordinates (and estimated frequencies) whose estimated frequencies are at least $\frac{\nu\eta T}{2}+X$, where $X$ is Laplacian noise with scale parameter $\O{\frac{1}{\eta^2\nu^2}}$. 
Then $\privcountsketch$ gives the following guarantees:
\begin{lemma}
\lemlab{lem:priv:cs}
There exists a one-pass streaming algorithm $\privcountsketch$ that takes an accuracy parameter $\nu\in(0,1)$ and a threshold parameter $\eta^2$ and outputs a list $H$ that contains all indices $k\in[n]$ of an underlying frequency vector $x$ with $x_k\ge\eta\,L_2(x)$ and no index $k\in[n]$ with $x_k\le\eta(1-\nu)\,L_2(x)$. 
For each $k\in H$, $\privcountsketch$ also reports a estimated frequency $\widehat{x_k}$ such that $(1-\nu)x_k-\O{\frac{\log m}{\eta\nu}}\le\widehat{x_k}\le(1+\nu)x_k+\O{\frac{\log m}{\eta\nu}}$. 
The algorithm uses $\O{\frac{1}{\eta^2\nu^2}\log^2 n}$ bits of space and succeeds with probability $1-\frac{1}{\poly(m)}$.
\end{lemma}

\subsection{Symmetric Norms}
\seclab{sec:prelims:sym}
In this section, we provide necessary preliminaries for symmetric norm estimation. 
For additional intuition, see \appref{app:intuition}. 

\begin{definition}[Symmetric norm]
A function $L:\mathbb{R}^n\to\mathbb{R}$ is a \emph{symmetric norm} if $L$ is a norm and for all $x\in\mathbb{R}^n$ and any vector $y\in\mathbb{R}^n$ that is a permutation of the coordinates of $x$, we have $L(x)=L(y)$. 
Moreover, we have $L(x)=L(|x|)$, where $|x|$ is the coordinate-wise absolute value of $x$. 
\end{definition}

\begin{definition}[Modulus of concentration]
Let $x\in\mathbb{R}^n$ be a random variable drawn from the uniform distribution on the $L_2$-unit sphere $S^{n-1}$ and let $b_L$ denote the maximum value of $L(x)$ over $S^{n-1}$. 
The median of a symmetric norm $L$ is the unique value $M_L$ such that $\PPr{L(x)\ge M_L}\ge\frac{1}{2}$ and $\PPr{L(x)\le M_L}\ge\frac{1}{2}$. 
Then the ratio $\mc(L):=\frac{b_L}{M_L}$ is the \emph{modulus of concentration} of the norm $L$. 
\end{definition}
Although the modulus of concentration quantifies the ``average'' behavior of the norm $L$ on $\mathbb{R}^n$, norms with challenging behavior can still be embedded in lower-dimensional subspaces. 
For instance, the $L_1$ norm satisfies $\mc(L)=\O{1}$, but when $x\in\mathbb{R}^n$ has fewer than $\sqrt{n}$ nonzero coordinates, the norm $\max(L_{\infty}(x),L_1(x)/\sqrt{n})$ on the unit ball becomes identically $L_{\infty}(x)$~\cite{BlasiokBCKY17}, which requires $\Omega(\sqrt{n})$ space~\cite{AlonMS99} to estimate. 
Hence, we further quantify the behavior of a norm $L$ by examining its behavior on all lower dimensions. 
\begin{definition}[Maximum modulus of concentration]
\deflab{def:mmc}
For a norm $L:\mathbb{R}^n\to\mathbb{R}$ and every $k\le n$, define the norm $L^{(k)}:\mathbb{R}^k\to\mathbb{R}$ by $L^{(k)}((x_1,\ldots,x_k)):=L((x_1,\ldots,x_k,0,\ldots,0))$.
Then the \emph{maximum modulus of concentration} of the norm $L$ is $\mmc(L):=\underset{k\le n}{\max}\mc(L^{(k)})=\underset{k\le n}{\max}\frac{b_{L^{(k)}}}{M_{L^{(k)}}}$.
\end{definition}

\begin{definition}[Important Levels]
\deflab{def:important:level}
For $x\in\mathbb{R}^n$ and $\xi>1$, we define the level $i$ as the set $B_i=\{k\in[n]\,:\,\xi^{i-1}\le|x_k|\le\xi^i\}$. 
We define $b_i:=|B_i|$ as the size of level $i$. 
For $\beta\in(0,1]$, we say level $i$ is \emph{$\beta$-important} if
\[b_i>\beta\sum_{j>i}b_j,\qquad b_i\xi^{2i}\ge\beta\sum_{j\le i}b_j\xi^{2j}.\]
\end{definition}
Informally, level $i$ is $\beta$-important if (1) its size is at least a $\beta$-fraction of the total sizes of the higher levels and (2) its contribution is roughly a $\beta$-fraction of the total contribution of all the lower levels. 
We would like to show that to approximate a symmetric norm $L(x)$, it suffices to identify the $\beta$-important levels and their sizes for a fixed base $\xi>1$. 
\begin{definition}[Level Vectors and Buckets]
\deflab{def:level:vector}
For $x\in\mathbb{R}^n$ and $\xi>1$, the \emph{level vector} for $x$ is
\begin{align*}
V(x):=&(\underbrace{\xi^1,\ldots,\xi^1}_{b_1\text{ times}},\underbrace{\xi^2,\ldots,\xi^2}_{b_2\text{ times}},\ldots,\underbrace{\xi^k,\ldots,\xi^k}_{b_k\text{ times}},0,\ldots,0)\in\mathbb{R}^n,
\end{align*}
where each $b_i$ is the size of level $i$. 
The $i$-th bucket of $V(x)$ is
\begin{align*}
V_i(x):=&(\underbrace{0,\ldots,0,}_{b_1+\ldots+b_{i-1}\text{ times}}\underbrace{\xi^i,\ldots,\xi^i}_{b_i\text{ times}},\ldots,\underbrace{0,\ldots,0}_{b_{i+1}+\ldots+b_k\text{ times}},0,\ldots,0)\in\mathbb{R}^n.
\end{align*}
We similarly define the approximate level vectors $\widehat{V(x)}$ and $\widehat{V_i(x)}$ using approximations $\widehat{b_1},\ldots,\widehat{b_k}$ for $b_1,\ldots,b_k$. 
We write $V(x)\setminus V_i(x)$ to denote the vector that replaces the $i$-th bucket in $V(x)$ with all zeros and we write $V(x)\setminus V_i(x)\cup\widehat{V_i(x)}$ to denote the vector that replaces the $i$-th bucket in $V(x)$ with $\widehat{b_i}$ instances of $\xi^i$. 
\end{definition}
Rather than directly handle the important levels, we define the $\beta$-contributing levels and instead work toward estimating the contribution of the $\beta$-contributing levels. 
\begin{definition}[Contributing Levels]
\deflab{def:contributing:level}
Given $x\in\mathbb{R}^n$, a level $i$ defined by base $\xi>1$ is \emph{$\beta$-contributing} if $L(V_i(x))\ge\beta L(V(x))$.  
\end{definition}
\cite{BlasiokBCKY17} showed that even if all levels that are not $\beta$-contributing are removed, the contribution of the remaining levels forms a good approximation to $L(x)$. 
\begin{lemma}\cite{BlasiokBCKY17}
Given $x\in\mathbb{R}^n$ and levels defined by a base $\xi>1$, let $V'(x)$ be the vector obtained by removing all levels that are not $\beta$-contributing from $V(x)$. 
Then $(1-\O{\log_{\xi}n}\cdot\beta)L(V(x))\le L(V'(x))\le L(V(x))$. 
\end{lemma}
Hence for appropriate $\xi>1$ and $\beta\in(0,1]$, it suffices to identify the $\beta$-contributing levels, zero out the remaining levels, and determine the contribution of the resulting vector to approximate the symmetric norm $L(x)$. 
\begin{lemma}
\lemlab{lem:reconstruction}
\cite{BlasiokBCKY17}
Given an accuracy parameter $\alpha\in(0,1]$, let base $\xi=(1+\O{\alpha})$, importance parameter $\beta=\O{\frac{\alpha^5}{\mmc(\ell)^2\cdot\log^5 m}}$, and $\alpha'=\O{\frac{\alpha^2}{\log n}}$. 
Let $\widehat{b_i}\le b_i$ for all $i$ and $\widehat{b_i}\ge(1-\alpha')b_i$ for all $\beta$-important levels. 
Let $\widehat{V}$ be the level vector constructed using the estimates $\widehat{b_1},\widehat{b_2},\ldots$ and let $V'$ be the level vector constructed by removing all the buckets that are not $\beta$-contributing in $\widehat{V}$. 
Then $(1-\alpha)L(V(x))\le L(V'(x))\le L(V(x))$. 
\end{lemma}
To identify the $\beta$-contributing levels, \cite{BlasiokBCKY17} first notes that the size of the level must be at least a significant fraction of the total size of the higher levels. 
\begin{lemma}
\lemlab{lem:contribute:num}
\cite{BlasiokBCKY17}
Given $x\in\mathbb{R}^n$, let the level sets be defined by a base $\xi>1$. 
If level $i$ is $\beta$-contributing, then there exists some fixed constant $\lambda>0$ such that 
\[b_i\ge\frac{\lambda\beta^2}{\mmc(\ell)^2\log^2 n}\cdot\sum_{j>i}b_j.\] 
\end{lemma}
Moreover, \cite{BlasiokBCKY17} observes that the squared mass of a $\beta$-contributing level must be at least a significant fraction of the total squared mass of the lower levels. 
\begin{lemma}
\lemlab{lem:contribute:weight}
\cite{BlasiokBCKY17}
Given $x\in\mathbb{R}^n$, let the level sets be defined by a base $\xi>1$. 
If level $i$ is $\beta$-contributing, then there exists some fixed constant $\lambda>0$ such that 
\[b_i\xi^{2i}\ge\frac{\lambda\beta^2}{\mmc(\ell)^2(\log_{\xi} n)\log^2 n}\cdot\sum_{j\le i}b_j\xi^{2j}.\] 
\end{lemma}
Observe that together, \lemref{lem:contribute:num} and \lemref{lem:contribute:weight} imply that a $\beta$-contributing level $i$ must also be an important level as defined in \defref{def:important:level}. 
Crucially, since \lemref{lem:contribute:weight} states that the squared mass (or the $F_2$ frequency moment) of the $\beta$-contributing levels must be a significant fraction of the total squared mass of the lower levels, then it suggests we might be able to identify the $\beta$-contributing levels through an $L_2$-heavy hitters algorithm after removing the higher levels. 
Indeed, \cite{BlasiokBCKY17} show that the problem of identifying the size (and thus the contribution) of the $\beta$-contributing levels can be reduced to the task of finding $\nu$-approximate $\eta$-heavy hitters for specific parameters of $\nu$ and $\eta$. 

\begin{lemma}
\lemlab{lem:beta:detect}
\cite{BlasiokBCKY17}
Let $s=\O{\log n}$. 
If a level $i$ is $\beta$-important, then either $\xi^{2i}\ge\frac{\alpha^2\beta\eps^2}{\log^2 m}\,F_2(x)$ or there exists $j\in[s]$ such that $b_i\ge\frac{2^j\log^2 m}{\alpha^2\eps^2}$ and $\xi^{2i}\in\left[\frac{\alpha^2\beta\eps^2}{\log^2 m}\cdot\frac{F_2(x)}{2^j},\frac{\alpha^2\beta\eps^2}{\log^2 m}\cdot\frac{F_2(x)}{2^{j-1}}\right]$. 
\end{lemma}
\lemref{lem:beta:detect} implies that if level $i$ is $\beta$-important, then either (1) it will be identified by using $\privcountsketch$, i.e., \lemref{lem:priv:cs}, with threshold $\frac{\alpha^2\beta}{\log^2 m}$ on the stream or (2) its contribution can be well-approximated by using $\privcountsketch$ with threshold $\frac{\alpha^2\beta\eps^2}{\log^2 m}$ on a substream formed by sampling coordinates of the universe with probability $\frac{1}{2^j}$. 
We thus split our algorithm and analysis to handle these cases. 
In particular, we call a frequency level $i$ ``high'' if $\xi^{2i}\ge\frac{\alpha^2\beta\eps^2}{\log^2 m}\,F_2(x)$. 
We call a frequency level $i$ ``medium'' if $\xi^{2i}\ge\frac{\alpha^2\beta'\eps^2}{2^j}\,F_2(x)>T$ and $b_i\ge\O{\frac{2^j\log^2 m}{\alpha^2\eps^2}}$ for a certain $\beta'>0$ and a threshold $T$. 
We call a frequency level $i$ ``low'' if $\xi^{2i}\ge\frac{\alpha^2\beta'\eps^2}{2^j}\,F_2(x)$ and $b_i\ge\O{\frac{2^j\log^2 m}{\alpha^2\eps^2}}$, but $T\ge\frac{\alpha^2\beta'\eps^2}{2^j}\,F_2(x)$. 

\section{Private Symmetric Norm Estimation Algorithm}
\seclab{sec:alg}
In this section, we give our algorithm that releases a set of private statistics from which an arbitrary number of symmetric norms can be well-approximated. 
In particular, recall that \lemref{lem:reconstruction} suggests that it suffices to approximate the sizes of the important levels and identity the non-important levels, so that the contributions of the non-important levels can be set to zero. 
We partition the levels into three groups after defining explicit thresholds $T_1$ and $T_2$, with $T_1>T_2$. 
Recall that we define the ``high frequency levels'' as the levels whose coordinates exceed $T_1$ in frequency, the ``medium frequency levels'' as the levels whose coordinates are between $T_1$ and $T_2$ in frequency, and the ``low frequency levels'' as the levels whose coordinates are less than $T_2$ in frequency. 

The intuition is that because the high frequency levels have such large magnitude, their frequencies can be well-approximated by running an $L_2$-heavy hitters algorithm on the stream $S$. 
On the other hand, the medium frequency level coordinates are not large enough to be detected by running an $L_2$-heavy hitters algorithm on the stream $S$, but the sizes of these level sets must be large if the level set is important and therefore, there exists a substream $S_j$ for which a large number of these coordinates are subsampled and their frequencies can be well-approximated by running an $L_2$-heavy hitters algorithm on the substream $S_j$. 
Here we form substreams $S_0,S_1,\ldots$ so that $S_j$ first samples elements of the universe $[n]$ at a rate $\frac{1}{2^j}$ and then only contains the stream updates that are relevant to the sampled elements. 
Finally, the low frequency level coordinates are small enough that we cannot add Laplacian noise to their frequencies without affecting the level sets they are mapped to. 
We instead show that $L_1$ sensitivity for the level set estimations is particularly small for the low frequency levels and thus, we report the size of the level sets of the low frequency levels rather than the approximate frequencies of the heavy-hitters. 

We emphasize that we only use the thresholds $T_1$ and $T_2$ for the purposes of describing our algorithm -- in the actual implementation of the algorithm, the thresholds $T_1$ and $T_2$ will be implicitly defined by each of the substreams. 
For example, the items with threshold larger than $T_1$ will automatically be revealed through the stream $S$, while the items with thresholds between $T_1$ and $T_2$ will be revealed through the substreams $S_j$ with $2^j>\frac{\log n}{\beta'\alpha\eps}$ for explicit parameters $\alpha$, $\beta'$, and $\eps$. 
More specifically, note that \algref{alg:hh:high} sets $\beta'=\mathcal{O}\left(\frac{\alpha^2\beta\varepsilon^2}{\log^2 m}\right)$ or more specifically $\beta'=\frac{\alpha^2\beta\varepsilon^2}{2\log^2 m}$. 
Then $\beta'\cdot F_2(x)$ corresponds to the threshold $T_1$, which is utilized in the proofs of \secref{sec:alg:high}. 
Similarly, \algref{alg:hh:med} leverages the quantity $\frac{\log n}{\beta'\alpha\varepsilon}$ to define the threshold $T_2$, which is then utilized in the proofs of \secref{sec:alg:med}. 

\subsection{Recovery of High Frequency Levels}
\seclab{sec:alg:high}
In this section, we describe our algorithm for recovering the high frequency levels, whose coordinates have sufficiently large magnitude and thus their frequencies can be well-approximated by running an $L_2$-heavy hitters algorithm on the stream $S$. 
Moreover, with high probability, adding Laplacian noise will not affect the level sets because the frequencies are so large. 
Thus it simply suffices to return the noisy estimated frequencies of each of the elements in the high frequency levels. 
This algorithm is the simplest of our cases and we give the algorithm in full in \algref{alg:hh:high}. 

\begin{algorithm}[!htb]
\caption{Algorithm to privately estimate the high levels}
\alglab{alg:hh:high}
\begin{algorithmic}[1]
\Require{Privacy parameter $\eps>0$, accuracy parameter $\alpha\in(0,1)$}
\Ensure{Private estimation of the frequencies of the coordinates of the high frequency levels}
\State{$\beta\gets\O{\frac{\alpha^5}{\mmc(L)^2\log^5 m}}$, $\beta'\gets\O{\frac{\alpha^2\beta\eps^2}{\log^2 m}}$}
\State{Run $\privcountsketch$ on the stream $S$ with threshold $\alpha^2\beta'$ and failure probability $\frac{1}{\poly(m)}$}
\For{each heavy-hitter $k\in[n]$ reported by $\privcountsketch$}
\State{Let $\widetilde{x_k}$ be the frequency estimated by $\privcountsketch$}
\State{$\widehat{x_k}\gets\widetilde{x_k}+\Lap\left(\frac{8}{\beta'\eps}\right)$}
\State{\Return $\widehat{x_k}$}
\EndFor
\end{algorithmic}
\end{algorithm}
We first show that coordinates in high frequency levels are identified and their frequencies are accurately estimated. 
\begin{lemma}
\lemlab{lem:high:correct:yes}
Suppose $x_k^2\ge\frac{\alpha^2\beta\eps^2}{\log^2 m}\,F_2(x)$ and $m=\frac{\Omega\left(\log^5 m\right)}{\alpha^5\beta^2\eps^5}$.  
Then with high probability, \algref{alg:hh:high} outputs $\widehat{x_k}$ such that 
\[(1-\alpha^2)x_k\le\widehat{x_k}\le x_k.\]
\end{lemma}
\begin{proof}
Consider \algref{alg:hh:high}. 
Since $x_k^2\ge\frac{\alpha^2\beta\eps^2}{2\log^2 m}\,F_2(x)$ and we call $\privcountsketch$ with threshold $\alpha^2\beta'$ with $\beta':=\O{\frac{\alpha^2\beta\eps^2}{\log^2 m}}$, then with high probability, the output $\widetilde{x_k}$ satisfies
\[(1-\O{\alpha^2})x_k\le\widetilde{x_k}\le x_k.\]
We then add Laplacian noise $\Lap\left(\frac{8}{\beta'\eps}\right)$ to $\widetilde{x_k}$ to form $\widehat{x_k}$. 
Since $x_k^2\ge\frac{\alpha^2\beta\eps^2}{2\log^2 m}\,F_2(x)=\beta'\,F_2(x)$ and $F_2(x)\ge m$, then with high probability, the Laplacian noise is at most an $\alpha^2$ fraction of $\widehat{x_k}$ for $\frac{\O{\log m}}{\beta'\eps}\le\alpha^2 m$ or equivalently, $m\ge\frac{\Omega\left(\log m\right)}{\alpha(\beta')^2\eps}\ge\frac{\Omega\left(\log^5 m\right)}{\alpha^5\beta^2\eps^5}$. 
Hence with high probability,
\[(1-\alpha^2)x_k\le\widehat{x_k}\le x_k.\]
\end{proof}
Similarly, we show that if a coordinate does not have high frequency, it will not be output by \algref{alg:hh:high}. 
\begin{lemma}
\lemlab{lem:high:correct:no}
Suppose $x_k^2<\frac{\alpha^2\beta\eps^2}{2\log^2 m}\,F_2(x)$ and $m=\frac{\Omega\left(\log^5 m\right)}{\alpha^5\beta^2\eps^5}$.  
Then with high probability, \algref{alg:hh:high} outputs $\widehat{x_k}$ such that 
\[\widehat{x_k}<\frac{3\alpha^2\beta\eps^2}{4\log^2 m}\,F_2(x).\]
\end{lemma}
\begin{proof}
Since $x_k^2<\frac{\alpha^2\beta\eps^2}{2\log^2 m}\,F_2(x)$ and we call $\privcountsketch$ with threshold $\alpha^2\beta'$ with $\beta':=\O{\frac{\alpha^2\beta\eps^2}{\log^2 m}}$, then the output $\widetilde{x_k}$ satisfies
\[|(\widetilde{x_k})^2-(x_k)^2|\le2\alpha^2\beta'\,F_2(x).\]
We then add Laplacian noise $\Lap\left(\frac{8}{\beta'\eps}\right)$ to $\widetilde{x_k}$ to form $\widehat{x_k}$. 
Since $F_2(x)\ge m$, then with high probability, the Laplacian noise is at most an $\alpha^2\beta'$ fraction of $F_2(x)$ for $\frac{\O{\log m}}{\beta'\eps}\le\alpha^2 m$ or equivalently, $m\ge\frac{\Omega\left(\log m\right)}{\alpha(\beta')^2\eps}\ge\frac{\Omega\left(\log^5 m\right)}{\alpha^5\beta^2\eps^5}$. 
Hence with high probability,
\[|(\widetilde{x_k})^2-(x_k)^2|\le\frac{\alpha^2\beta\eps^2}{4\log^2 m}\,F_2(x).\]
Since $x_k^2<\frac{\alpha^2\beta\eps^2}{2\log^2 m}\,F_2(x)$, then it follows that 
\[\widehat{x_k}<\frac{3\alpha^2\beta\eps^2}{4\log^2 m}\,F_2(x).\]
\end{proof}
\noindent
We now show that \algref{alg:hh:high} preserves differential privacy.
\begin{lemma}
\lemlab{lem:priv:high}
\algref{alg:hh:high} is $\left(\frac{\eps}{4},\frac{\delta}{4}\right)$-differentially private for $\delta=\frac{1}{\poly(m)}$. 
\end{lemma}
\begin{proof}
By \lemref{lem:sens:countsketch}, the sensitivity of $\privcountsketch$ is at most $2$ and the failure probability is $\frac{1}{\poly(m)}$. 
Thus by adding Laplacian noise $\Lap\left(\frac{8}{\beta'\eps}\right)$ to $\widetilde{x_k}$, each estimated frequency is $\left(\frac{\beta'\eps}{4},\frac{\delta}{4\beta}\right)$-differentially private for $\delta=\frac{1}{\poly(m)}$. 
Since $\privcountsketch$ with threshold $\beta'$ can release at most $\frac{1}{\beta}$ estimated frequencies and post-processing does not cause loss in privacy, then by \thmref{thm:dp:comp}, \algref{alg:hh:high} is $\left(\frac{\eps}{4},\frac{\delta}{4}\right)$. 
\end{proof}
Finally, we analyze the space complexity of \algref{alg:hh:high}. 
\begin{lemma}
\lemlab{lem:space:high}
\algref{alg:hh:high} uses space $\mmc(L)^2\cdot\poly\left(\frac{1}{\alpha},\frac{1}{\eps},\log m\right)$. 
\end{lemma}
\begin{proof}
The space complexity follows from running a single instance of $\privcountsketch$ with threshold $\alpha^2\beta'$ and failure probability $\frac{1}{\poly(m)}$, where $\beta'=\O{\frac{\alpha^2\beta\eps^2}{\log^2 m}}$ and $\beta=\O{\frac{\alpha^5}{\mmc(L)^2\log^5 m}}$. 
\end{proof}

\subsection{Recovery of Medium Frequency Levels}
\seclab{sec:alg:med}
In this section, we describe our algorithm for recovering the medium frequency levels, whose coordinates do not have sufficiently large magnitude to be detected by running an $L_2$-heavy hitters algorithm on the stream $S$, but have sufficiently large size, so that there exists some $j\in[s]$ across the $s$ subsampling levels such that the coordinates can be detected by running an $L_2$-heavy hitters algorithm on the stream $S_j$. 
On the other hand, their magnitudes are sufficiently large so that with high probability, adding Laplacian noise will not affect the level sets. 
We give the algorithm in full in \algref{alg:hh:med}. 

\begin{algorithm}[!htb]
\caption{Algorithm to privately estimate the medium levels}
\alglab{alg:hh:med}
\begin{algorithmic}[1]
\Require{Privacy parameter $\eps>0$, accuracy parameter $\alpha\in(0,1)$}
\Ensure{Private estimations of the sizes of the medium frequency levels}
\State{$\beta\gets\O{\frac{\alpha^5}{\mmc(L)^2\log^5 m}}$, $\beta'\gets\O{\frac{\alpha^3\beta\eps^2}{\log^2 m}}$, $\xi\gets(1+\O{\eps})$}
\State{$\gamma\gets(1/2,1)$ uniformly at random, $\ell\gets\ceil{\log_{\xi}(2m)}$, $s\gets\O{\log n}$}
\For{$j\in[s]$ with $2^j>\frac{\log n}{\beta'\alpha\eps}$}
\State{Form stream $S_j$ by sampling elements of $[n]$ with probability $\frac{1}{2^j}$}
\State{Run $\privcountsketch_j$ on stream $S_j$ with threshold $\alpha^2\beta'\eps^2$ and failure probability $\frac{1}{\poly(m)}$}
\For{each heavy-hitter $k\in[n]$ reported by $\privcountsketch_j$}
\State{Let $\widehat{x_k}$ be the frequency estimated by $\privcountsketch_j$}
\If{$\widehat{x_k}>\frac{\log n}{\beta'\alpha\eps}$}
\State{$\widetilde{x_k}\gets\widehat{x_k}+\Lap\left(\frac{8}{\beta'\eps}\right)$}
\EndIf
\EndFor
\For{$i\in[\ell]$ with $\frac{m^2}{2^{j+1}}>\gamma\xi^{2i}\ge 2^j>\O{\frac{\log n}{\beta'\alpha^2\eps}}$}
\State{Let $\widetilde{b_i}$ be the number of indices $k\in[n]$ such that $\gamma\xi^{2i}\le\widetilde{x_k}<\gamma\xi^{2i+2}$}
\State{$\widehat{b_i}\gets\frac{2^j}{(1+\O{\alpha})}\,\widetilde{b_i}$}
\State{\Return $\widehat{b_i}$}
\EndFor
\EndFor
\end{algorithmic}
\end{algorithm}
We first upper bound the second frequency moment (and hence the $L_2$ norm) of each substream. 
This is necessary because we want to detect the coordinates of the medium frequency levels as $L_2$-heavy hitters for each substream, but if the substream has overwhelmingly large $L_2$ norm, then we will not be able to find coordinates of the medium frequency levels. 
However, it may not be true that $F_2(S_j)$ is significantly smaller than $F_2(S)$ with high probability. 
For example, if there were a single large element, then the probability it is sampled at level $s$ is $\frac{1}{2^s}$, which is roughly $\frac{1}{n}>\frac{1}{\poly(m)}$. 
Instead, we note that $\privcountsketch$ benefits from the stronger \emph{tail guarantee}, which states that not only does $\privcountsketch$ with threshold $\eta<1$ detect the elements $k$ such that $(x_k)^2\ge\eta F_2(S)$, but it also detects the elements $k$ such that $(x_k)^2\ge\eta F_2(S_{\tail(1/\eta)})$, where $S_{\tail(1/\eta)}$ is the frequency vector $x$ induced by $S$, with the largest $\frac{1}{\eta}$ entries instead set to zero~\cite{BravermanCINWW17,BravermanGLWZ18}. 
\begin{lemma}
\lemlab{lem:moment:approx}
With high probability, we have that $F_2((S_j)_{1/(\alpha^2\beta'\eps^2)})\le\frac{200\log m}{2^j}\,F_2(x)$ for all $j\in[s]$. 
\end{lemma}
\begin{proof}
For each $j\in[s]$, we have that $\Ex{F_2(S_j)}=\frac{F_2(x)}{2^j}$. 
By Chernoff bounds with $\O{\log n}$-wise limited independence, we have that
\[\PPr{F_2((S_j)_{1/(\alpha^2\beta'\eps^2)})>\frac{200\log m}{2^j}\,F_2(x)}\le\frac{1}{\poly(m)}.\]
Since $s\le2\log m$, then by a union bound over all $j\in[s]$, we have that $F_2(S_j)\le (200\log m)F_2(x)$ for all $j\in[s]$. 
\end{proof}
We now show that conditioned on the event that the $L_2$ norm of the subsampled streams are not too large, then we can well-approximate the frequency of any coordinate of the medium frequency levels, provided that they are sampled in the substream. 
\begin{lemma}
\lemlab{lem:med:freq:approx}
Suppose $i$ is a $\beta$-important level and $k\in[n]$ is in level $i$, so that $x_k\in[\xi^i,\xi^{i+1})$. 
If $F_2((S_j)_{1/(\alpha^2\beta'\eps^2)})\le\frac{200\log m}{2^j}\,F_2(x)$ for all $j\in[s]$ and $k$ is sampled in stream $S_j$ with $2^j>\frac{\log n}{\beta'\alpha\eps}$, then with high probability, \algref{alg:hh:med} outputs $\widehat{x_k}$ such that 
\[(1-\alpha^2)x_k\le\widehat{x_k}\le x_k.\]
\end{lemma}
\begin{proof}
Consider \algref{alg:hh:med}. 
By \lemref{lem:beta:detect}, $x_2^2\in\left[\frac{\alpha^2\beta\eps^2}{\log^2 m}\cdot\frac{F_2(x)}{2^j},\frac{\alpha^2\beta\eps^2}{\log^2 m}\cdot\frac{F_2(x)}{2^{j-1}}\right]$. 
Conditioned on the event that $F_2((S_j)_{1/(\alpha^2\beta'\eps^2)})\le\frac{200\log m}{2^j}\,F_2(x)$ for all $j\in[s]$, then $x_k^2\ge\frac{\alpha^2\beta\eps^2}{200\log m}\,F_2(S_j)$. 
We call $\privcountsketch$ with threshold $\alpha^2\beta'\eps^2=\O{\frac{\alpha^4\beta\eps^3}{\log^2 m}}$. 
Thus with high probability, the output $\widetilde{x_k}$ satisfies
\[(1-\O{\alpha^2})x_k\le\widetilde{x_k}\le x_k.\]
We then add Laplacian noise $\Lap\left(\frac{8}{\beta'\eps}\right)$ to $\widetilde{x_k}$ to form $\widehat{x_k}$. 
Since $x_k^2\ge\O{\frac{\log n}{\beta'\alpha^2\eps}}$, then with high probability, the Laplacian noise is at most an $\alpha^2$ fraction of $\widehat{x_k}$. 
Hence with high probability,
\[(1-\alpha^2)x_k\le\widehat{x_k}\le x_k.\]
\end{proof}
Unfortunately, \lemref{lem:med:freq:approx} only provides guarantees for the coordinates of the medium frequency levels that are sampled. 
Thus, we still need to use \lemref{lem:med:freq:approx} to show that a good estimator to the sizes of the medium frequency levels can be obtained from the estimates of the coordinates of the medium frequency levels that are sampled. 
In particular, we show that rescaling the empirical sizes of the medium frequency levels forms a good estimator to the actual sizes of the medium frequency levels. 
\begin{lemma}
\lemlab{lem:med:levels:approx}
Consider a $\beta$-important level $i$ with $\xi^{2i}\in\left[\frac{\beta\alpha^2\eps^2}{\log^2 m}\cdot\frac{F_2(x)}{2^j},\frac{\beta\alpha^2\eps^2}{\log^2 m}\cdot\frac{F_2(x)}{2^{j-1}}\right]$ for some integer $j>0$ and $\xi^i>\frac{\log n}{\beta'\alpha\eps}$. 
If $F_2((S_j)_{1/(\alpha^2\beta'\eps^2)})\le\frac{200\log m}{2^j}\,F_2(x)$ for all $j\in[s]$, then with high probability, \algref{alg:hh:med} outputs $\widehat{b_i}$ such that 
\[(1-\O{\alpha})b_i\le\widehat{b_i}\le b_i,\]
where $b_i$ is the size of level $i$. 
\end{lemma}
\begin{proof}
Suppose $i$ is a $\beta$-important level. 
Then by \lemref{lem:beta:detect} and a shifting of the index $j$, $b_i\ge\O{\frac{2^j\log^2 m}{\alpha^2\eps^2}}$. 
Thus in $S_j$, the expected number of items $E_j$ from level $i$ is at least $\frac{\log^2 m}{\alpha^2\eps^2}$ and the variance $V_j$ is at most $E_j$. 
Hence by Chernoff bounds with $\O{\log n}$-wise limited independence, we have that the number of items $N_j$ from level $i$ satisfies 
\[(1-\O{\alpha})b_i\le 2^j\cdot N_j\le(1+\O{\alpha})b_i,\]
with high probability. 
\cite{BlasiokBCKY17} show that due to the uniformly random chosen $\gamma\in(1/2,1)$, we further have
\[(1-\O{\alpha})N_j\le(1+\O{\alpha})\widehat{b_i}\le(1+\O{\alpha})N_j,\]
with high probability. 
Since $s\le2\log m$, then by a union bound over all $j\in[s]$, we have that with high probability, \algref{alg:hh:med} outputs $\widehat{b_i}$ such that 
\[(1-\O{\alpha})b_i\le\widehat{b_i}\le b_i.\]
\end{proof}
We now show that \algref{alg:hh:med} preserves differential privacy.
\begin{lemma}
\lemlab{lem:priv:med}
\algref{alg:hh:med} is $\left(\frac{\eps}{4},\frac{\delta}{4}\right)$-differentially private for $\delta=\frac{1}{\poly(m)}$. 
\end{lemma}
\begin{proof}
By \lemref{lem:sens:countsketch}, the sensitivity of $\privcountsketch$ is at most $2$ and the failure probability is $\frac{1}{\poly(m)}$. 
Thus by adding Laplacian noise $\Lap\left(\frac{8}{\beta'\eps}\right)$ to $\widetilde{x_k}$, each estimated frequency is $\left(\frac{\beta'\eps}{4},\frac{\delta}{4\beta}\right)$-differentially private for $\delta=\frac{1}{\poly(m)}$. 
Since $\privcountsketch$ with threshold $\beta'$ can release at most $\frac{1}{\beta}$ estimated frequencies, then by \thmref{thm:dp:comp}, \algref{alg:hh:med} is $\left(\frac{\eps}{4},\frac{\delta}{4}\right)$. 
\end{proof}
It remains to analyze the space complexity of \algref{alg:hh:med}. 
\begin{lemma}
\lemlab{lem:space:med}
\algref{alg:hh:med} uses space $\mmc(L)^2\cdot\poly\left(\frac{1}{\alpha},\frac{1}{\eps},\log m\right)$. 
\end{lemma}
\begin{proof}
The space complexity follows from running $s$ instances of $\privcountsketch$ with threshold $\alpha^2\beta'$ and failure probability $\frac{1}{\poly(m)}$, where $\beta'=\O{\frac{\alpha^2\beta\eps^2}{\log^2 m}}$ and $\beta=\O{\frac{\alpha^5}{\mmc(L)^2\log^5 m}}$. 
Since $s=\O{\log n}$ and we assume $n\le m$ so that $\O{\log n}=\O{\log m}$, then the space complexity follows. 
\end{proof}

\subsection{Recovery of Low Frequency Levels}
\seclab{sec:alg:low}
In this section, we describe our algorithm for recovering the low frequency levels, whose coordinates have magnitude small enough that we cannot add Laplacian noise to their frequencies without affecting the corresponding level set sizes. 
We instead report the sizes of the level sets for the low frequency levels rather than the approximate frequencies of the heavy-hitters. 
Thus we must add Laplacian noise to the sizes of the level sets; we show that the $L_1$ sensitivity for the level set estimations is particularly small for the low frequency levels and thus the Laplacian noise does not greatly affect the estimates of the level set sizes. 
We note that this approach does not work for the high frequency levels because the high frequency levels may have small level set sizes, so that adding Laplacian noise to the sizes can significantly affect the resulting estimates of the level set sizes. 
Similarly, it is more challenging to argue the low $L_1$ sensitivity for the level set estimations for the medium frequency levels. 
Hence, both the algorithm and analysis are especially well-catered to the low frequency levels. 
We give the algorithm in full in \algref{alg:hh:low}. 

\begin{algorithm}[!htb]
\caption{Algorithm to privately estimate the low levels}
\alglab{alg:hh:low}
\begin{algorithmic}[1]
\Require{Privacy parameter $\eps>0$, accuracy parameter $\alpha\in(0,1)$}
\Ensure{Private estimations of the sizes of the low frequency levels}
\State{$\beta\gets\O{\frac{\alpha^5}{\mmc(L)^2\log^5 m}}$, $\beta'\gets\O{\frac{\alpha^2\beta\eps}{\log n}}$, $\xi\gets(1+\O{\eps})$}
\State{$\gamma\gets(1/2,1)$ uniformly at random, $\ell\gets\ceil{\log_{\xi}(2m)}$, $s\gets\O{\log n}$}
\For{$j\in[s]$ with $2^j\le\frac{\log n}{\beta'\alpha\eps}$}
\State{Form stream $S_j$ by sampling elements of $[n]$ with probability $\frac{1}{2^j}$}
\State{Run $\privcountsketch_j$ on stream $S_j$ with threshold $\beta'':=\O{\frac{\beta'\alpha^2\eps^3}{\log^2 n}}$}
\For{each heavy-hitter $k\in[n]$ reported by $\privcountsketch_j$}
\State{Let $\widehat{x_k}$ be the frequency estimated by $\privcountsketch_j$}
\EndFor
\For{$i\in[\ell]$ with $\O{\frac{\log n}{\beta'\alpha^2\eps}}\ge 2^{j+1}>\gamma\xi^{2i}\ge 2^j$}
\State{Let $\widetilde{b_i}$ be the number of indices $k\in[n]$ such that $\gamma\xi^{2i}\le\widehat{x_k}<\gamma\xi^{2i+2}$}
\State{$\widehat{b_i}\gets\frac{2^j}{(1+\O{\alpha})}\left(\widetilde{b_i}+\Lap\left(\frac{8}{\eps}\right)\right)$}
\State{\Return $\widehat{b_i}$}
\EndFor
\EndFor
\end{algorithmic}
\end{algorithm}

We first show that the estimates of the level set sizes for the low frequency levels are accurate. 
\begin{lemma}
\lemlab{lem:low:levels:approx}
Consider a $\beta$-important level $i$ with $\xi^{2i}\in\left[\frac{\beta\alpha^2\eps^2}{\log^2 m}\cdot\frac{F_2(x)}{2^j},\frac{\beta\alpha^2\eps^2}{\log^2 m}\cdot\frac{F_2(x)}{2^{j-1}}\right]$ for some integer $j>0$ and $\xi^i\le\frac{\log n}{\beta'\alpha\eps}$. 
If $F_2((S_j)_{1/(\alpha^2\beta'\eps^2)})\le\frac{200\log m}{2^j}\,F_2(x)$ for all $j\in[s]$, then with high probability, \algref{alg:hh:low} outputs $\widehat{b_i}$ such that 
\[(1-\O{\alpha})b_i\le\widehat{b_i}\le b_i,\]
where $b_i$ is the size of level set $i$. 
\end{lemma}
\begin{proof}
Suppose $i$ is a $\beta$-important level. 
Hence by a shifting of the index $j$ in \lemref{lem:beta:detect}, we have that $b_i\ge\O{\frac{2^j\log^2 m}{\alpha^2\eps^2}}$. 
Therefore, the expected number of items $E_j$ from level $i$ sampled in the substream $S_j$ is at least $\frac{\log^2 m}{\alpha^2\eps^2}$ and the variance $V_j$ is at most $E_j$. 
Thus by Chernoff bounds with $\O{\log n}$-wise limited independence, the number of items $N_j$ from level $i$ satisfies 
\[(1-\O{\alpha})b_i\le 2^j\cdot N_j\le(1+\O{\alpha})b_i,\]
with high probability. 
\cite{BlasiokBCKY17} show that due to the uniformly random chosen $\gamma\in(1/2,1)$, we further have
\[(1-\O{\alpha})N_j\le(1+\O{\alpha})\widehat{b_i}\le(1+\O{\alpha})N_j,\]
with high probability. 
Since $s\le2\log m$ and $\Lap\left(\frac{8}{\eps}\right)$ is at most an $\eps$-fraction of $b_i\ge\O{\frac{2^j\log^2 m}{\alpha^2\eps^2}}$ with high probability, then by a union bound over all $j\in[s]$, we have that with high probability, \algref{alg:hh:med} outputs $\widehat{b_i}$ such that 
\[(1-\O{\alpha})b_i\le\widehat{b_i}\le b_i.\]
\end{proof}
We then show that \algref{alg:hh:low} is differentially private. 
\begin{lemma}
\lemlab{lem:priv:low}
\algref{alg:hh:low} is $\left(\frac{\eps}{4},\frac{\delta}{4}\right)$-differentially private for $\delta=\frac{1}{\poly(m)}$. 
\end{lemma}
\begin{proof}
Note that since each instance of $\privcountsketch_j$ uses threshold $\beta'':=\O{\frac{\beta'\alpha^2\eps^3}{\log^2 n}}$ on a stream $S_j$ with $F_2(S_j)\le\frac{200\log m}{2^j}\,F_2(x)$, then for any $k\in[n]$ with $x_k\le\O{\frac{\log n}{\beta'\alpha^2\eps}}$, we have that $\privcountsketch_j$ outputs $x_k$ exactly. 
Hence, at most two estimates of the sizes of the level sets $\widehat{b_i}$ can change, and then can change by at most one. 
Thus the sensitivity is at most $2$, so it suffices to add Laplcian noise $\Lap\left(\frac{8}{\eps}\right)$ to each estimate $\widehat{b_i}$ to obtain $\left(\frac{\eps}{4},\frac{\delta}{4}\right)$-differentially private for $\delta=\frac{1}{\poly(m)}$. 
\end{proof}
Finally, we argue the space complexity of \algref{alg:hh:low}. 
\begin{lemma}
\lemlab{lem:space:low}
\algref{alg:hh:low} uses space $\mmc(L)^2\cdot\poly\left(\frac{1}{\alpha},\frac{1}{\eps},\log m\right)$. 
\end{lemma}
\begin{proof}
Similar to \algref{alg:hh:med}, the space complexity follows as a result of running $s$ instances of $\privcountsketch$ with threshold $\alpha^2\beta'$ and failure probability $\frac{1}{\poly(m)}$, where $\beta'=\O{\frac{\alpha^2\beta\eps^2}{\log^2 m}}$ and $\beta=\O{\frac{\alpha^5}{\mmc(L)^2\log^5 m}}$. 
Since $s=\O{\log n}$ and we assume $n\le m$ so that $\O{\log n}=\O{\log m}$, then the space complexity follows. 
\end{proof}

\subsection{Putting Things Together}
We would like to combine the subroutines from the previous sections to output a private dataset for symmetric norm estimation. 
Thus it remains to describe how to privately partition the coordinates into the high, medium, and low frequency levels. 
To that end, we remark that by \lemref{lem:sens:countsketch}, the sensitivity of $\privcountsketch$ in \algref{alg:cs} is at most $2$. 
Moreover, although $\privcountsketch$ actually provides an estimated frequency for each coordinate, for our purposes, we only need estimated frequencies for the $L_2$-heavy hitters and there are at most $K:=\O{\frac{1}{\eta^2}}$ possible $L_2$-heavy hitters with whichever threshold $\eta$ that we choose, e.g., $\eta=\alpha^2\beta'$ in \algref{alg:hh:high}. 
Thus it suffices to observe that we can privately partition the coordinates into the high, medium, and low frequency levels by first privately outputting the top $K$ estimated frequencies and then partitioning the coordinates according to their noisy estimated frequencies, which can be viewed as post-processing. 
In particular,~\cite{QiaoSZ21} observes that it suffices to add Laplacian noise with scale $\frac{8}{\eta\eps}$ to each of the frequencies and then outputting the top $K$ noisy estimated frequencies to achieve $\frac{\eps}{4}$-differential privacy. 

We now finally put together the results from the previous sections to show the following result. 
We remark that we set $\eps,\alpha=\tilde{\Omega}\left(\left(\frac{M^2}{m}\right)^{\frac{1}{30}}\right)$ so that along with the assumption that $m\ge n$, the conditions of the previous statements are satisfied, e.g., \lemref{lem:high:correct:yes}, we obtain the following formalization of \thmref{thm:main}.

\begin{theorem}
Given a parameter $M>1$, let $\eps,\alpha=\tilde{\Omega}\left(\left(\frac{M^2}{m}\right)^{\frac{1}{30}}\right)$. 
There exists a $(\eps,\delta)$-differentially private algorithm that outputs a set $C$, for $\delta=\frac{1}{\poly(m)}$. 
From $C$, the $(1+\alpha)$-approximation to any norm with \emph{maximum modulus of concentration} at most $M$ can be computed, with probability at least $1-\delta$. 
The algorithm uses $M^2\cdot\poly\left(\frac{1}{\alpha},\frac{1}{\eps},\log m\right)$ bits of space. 
\end{theorem}
\begin{proof}
Note that from \lemref{lem:high:correct:yes} and \lemref{lem:high:correct:no}, the frequencies of the coordinates in the high frequency levels are well-approximated with high probability. 
Similarly, from \lemref{lem:med:levels:approx} and \lemref{lem:low:levels:approx}, the sizes of the level sets of the medium and low frequency levels are well-approximated  with high probability. 
Moreover, all the level sets are partitioned into the high, medium, or low frequency levels. 
We would like to say that by \lemref{lem:reconstruction}, these statistics are sufficient to recover a $(1+\alpha)$-approximation to any norm with \emph{maximum modulus of concentration} at most $M$ and so we achieve a $(1+\alpha)$-approximation to any norm with maximum modulus of concentration at most $M$ that with high probability. 
Indeed, in an idealized process where $\xi^i\le\widehat{x_k}\le\xi^{i+1}$ if and only if $k$ is sampled by the substream $j$ assigned to level $i$ and $\xi^i\le x_k<\xi^{i+1}$, \lemref{lem:reconstruction} would show that we achieve a $(1+\alpha)$-approximation to any norm with maximum modulus of concentration at most $M$ that with high probability.
However, this may not always be the case because the frequency $x_k$ may lie near the boundary of the interval $[\xi^i,\xi^{i+1})$ and the estimate $\widehat{x_k}$ may lie outside of the interval, in which case $\widehat{x_k}$ is used toward the estimation of some other level set. 
Thus, our algorithm randomizes the boundaries of the level sets by instead defining the level sets as $[\gamma\xi^i,\gamma\xi^{i+1})$ for some $\gamma\in(1/2,1)$ chosen uniformly at random. 
Since we call $\privcountsketch$ with threshold at most $\alpha^2\beta'$, then the probability that item $k\in[n]$ is misclassified over the choice of $\gamma$ is at most $\O{\alpha^2\beta'}$. 
Furthermore, if $k$ in level set $i$ is misclassified, it can only be classified into level set $i-1$ or $i+1$, causing at most an incorrect multiplicative factor of two. 
Then in expectation across all $k\in[n]$, the error due to the misclassification is at most an $\O{\alpha^2\beta'}$ fraction of the symmetric norm. 
Hence by Markov's inequality, the error due to the misclassification is at most an additive $\frac{\alpha}{2}$ fraction of the symmetric norm with probability at least $0.99$. 
To obtain high probability of success, it then suffices to take the median across $\O{\log m}$ independent instances, finally showing correctness of our algorithm. 

The private partitioning of the coordinates into the high, medium, and low frequency levels is $\frac{\eps}{4}$-differentially private.
Each of the three sets of statistics released by the high, medium, and low frequency levels are $\left(\frac{\eps}{4},\frac{\delta}{4}\right)$-differentially private, by \lemref{lem:priv:high}, \lemref{lem:priv:med}, and \lemref{lem:priv:low}. 
Then $(\eps,\delta)$-differential privacy follows from the composition of differential privacy, i.e., \thmref{thm:dp:comp}.

Finally, the space complexity follows from \lemref{lem:space:high}, \lemref{lem:space:med}, and \lemref{lem:space:low}. 
\end{proof}

We remark that our algorithm is presented as having unlimited access to random bits but is analyzed using $\O{\log m}$-wise independence, so it can be properly derandomized to provide the space guarantees without needing to store a large number of random bits. 
Alternatively, our algorithm can also be derandomized using Nisan’s pseudorandom generator, which induces an extra multiplicative factor of $\O{\log m}$ in the space overhead~\cite{Nisan92}. 

Finally, we remark that the failure probability can be raised from $\delta=\frac{1}{\poly(m)}$ to arbitrarily $\delta>0$ using additional space overhead $\polylog\frac{1}{\delta}$, since the space dependency in each subroutine on the failure probability $\delta$ is $\polylog\frac{1}{\delta}$. 

\begin{theorem}
\thmlab{thm:main:formal}
Given a parameter $M>1$, let $\eps,\alpha=\tilde{\Omega}\left(\left(\frac{M^2}{m}\right)^{\frac{1}{30}}\right)$. 
There exists a $(\eps,\delta)$-differentially private algorithm that outputs a set $C$, from which the $(1+\alpha)$-approximation to any norm , with \emph{maximum modulus of concentration} at most $M$ of a vector $x\in\mathbb{R}^n$ induced by a stream of length $\poly(n)$ can be computed, with probability at least $1-\delta$. 
The algorithm uses $M^2\cdot\poly\left(\frac{1}{\alpha},\frac{1}{\eps},\log n,\log\frac{1}{\delta}\right)$ bits of space. 
\end{theorem}

\def\shortbib{0}
\bibliographystyle{alpha}
\bibliography{references}

\appendix

\section{Additional Intuition}
\applab{app:intuition}
\subsection{Maximum Modulus of Concentration}
Let $x\in\mathbb{R}^n$ be a random variable drawn from the uniform distribution on the $L_2$-unit sphere $S^{n-1}$ and let $b_L$ denote the maximum value of $L(x)$ over $S^{n-1}$. 
The median of a symmetric norm $L$ is the unique value $M_L$ such that $\PPr{L(x)\ge M_L}\ge\frac{1}{2}$ and $\PPr{L(x)\le M_L}\ge\frac{1}{2}$. 
Then recall that the ratio $\mc(L):=\frac{b_L}{M_L}$ is the \emph{modulus of concentration} of the norm $L$. 

For a vector $x\in\mathbb{R}^n$, the $L_1$ norm is defined as $L_1(x)=\sum_{i=1}^n|x_i|$. 
The maximum value of $L_1(x)$ for a vector $x$ from the $L_2$-unit sphere $S^{n-1}$ is $L_1(x)=\sqrt{n}$ for the flat vector $\left(\frac{1}{\sqrt{n}},\ldots,\frac{1}{\sqrt{n}}\right)$. 
Thus we have $b_{L_1}=\sqrt{n}$. 
It turns the median value of $L_1(x)$ for a vector $x$ from the $L_2$-unit sphere $S^{n-1}$ is also $M_{L_1}=\tilde{\Omega}(\sqrt{n})$ and so $\mc(L_1)\le\polylog(n)$.  

On the other hand, for the $L_3$ norm defined as $L_3(x)=\left(\sum_{i=1}^n|x_i|^3\right)^3$, the maximum value of $L_3(x)$ for a vector $x$ from the $L_2$-unit sphere $S^{n-1}$ is $L_3(e_i)=1$ for a unit vector $e_i$., while the median value of $L_3(x)$ for a vector $x$ from the $L_2$-unit sphere $S^{n-1}$ is roughly $M_{L_3}=\O{n^{-1/6}}$ and so $\mc(L_3)=\Omega(n^{-1/6})$.  
Thus we should expect the complexity of the estimating the $L_3$ norm to be significantly more challenging than estimating the $L_1$ norm, and indeed this reflects known upper and lower bounds in the streaming model, e.g., see the discussion in~\cite{WoodruffZ21}. 

\subsection{\texorpdfstring{Intuition on \cite{IndykW05} and \cite{BlasiokBCKY17}}{Intuition on [IW05] and [BBC+17]}}
The main intuition of the celebrated Indyk-Woodruff norm estimation algorithm~\cite{IndykW05,BlasiokBCKY17} is to decompose a norm $\ell(x)$ on input vector $x$ into the contribution by each of its coordinates, which can then be partitioned into level sets, based on how much they contribute to the norm $\ell(x)$. 
We can then approximate each of the contributions of the level sets by subsampling the universe and estimating the sizes of each universe through the heavy-hitters of each subsample. 

Recall the definition of the important levels in \defref{def:important:level}. 
Intuitively, an important level if its size is ``significant'' compared to all the higher levels and its contribution is ``significant'' compared to all the lower levels, so that the important level contributes a ``significant'' amount to the overall norm, which is formalized by \lemref{lem:beta:detect} to be the $\beta$-contributing levels, i.e., see \defref{def:contributing:level}.  

The norm estimation algorithms of \cite{IndykW05,BlasiokBCKY17} then reconstructs an estimate of the level vector, i.e., see \defref{def:level:vector}, by removing all the levels that are not $\beta$-contributing and using a $(1+\eps')$-approximation to the sizes of all $\beta$-important levels, for some fixed value of $\eps'>0$, which is a function of the accuracy parameter $\eps$ and the maximum modulus of concentration $\mmc(\ell)$ of the norm $\ell$. 
\lemref{lem:beta:detect} then shows that this approach suffices to obtain a $(1+\eps)$-approximation to $\ell(x)$. 

In particular, we can use $\nu$-approximate $\eta$-heavy hitters algorithms to roughly estimate the size $b_i$ of all $\beta$-important levels, because each $\beta$-important level must have either (1) large size, i.e., a large number of coordinates achieving a certain range of frequencies, or (2) large contribution, i.e., a small number of coordinates with significantly large value or (3) both. 
If the $\beta$-important level has large contribution but small size, then the significantly large coordinates will immediately be recognized as heavy-hitters. 
Otherwise, if the $\beta$-important level has large size, then a large number of these coordinates in the level set will be subsampled. 
Then these sampled coordinates of the level set will become heavy-hitters at some level $i$ in which $\Theta\left(\frac{1}{\eps^2}\right)$ of these coordinates are subsampled, since the expected $\ell$ norm of the sampled coordinates will be also be significantly smaller. 
The size $b_i$ of these level sets with a large number of coordinates can be then roughly estimated by rescaling the number of sampled coordinates by the inverse of the sampling probability, though additional care is required to formalize this argument, e.g., randomized boundaries for each of the level sets. 
\end{document}